%% file: main.tex
\documentclass[acmsmall,screen]{acmart}
\pdfoutput=1
\setcopyright{none}
\usepackage[conf={no link to proof,restate,text proof=Proof}]{proof-at-the-end}

\include{prelude}
\pgfplotsset{compat=1.17}

\newcommand{\conv}[1]{\textit{conv}\left(#1\right)}

\newcommand{\transpose}[1]{#1^T}

\newcommand{\reduce}{\mathbf{red}}

\newcommand{\Qplus}{\mathbb{Q}^{\geq 0}}
\newcommand{\Zplus}{\mathbb{Z}^{\geq 0}}
\newcommand{\tuple}[1]{\left( #1 \right)}     
\newcommand{\product}[2]{#1\!\left\langle #2 \right\rangle}
\newcommand{\qideal}[1]{\left\langle #1 \right\rangle}
\newcommand{\cone}[1]{\product{\Qplus}{#1}}
\newcommand{\inthull}[1]{#1_I}
\newcommand{\floor}[1]{\left\lfloor#1\right\rfloor}
\newcommand{\cp}{\textit{cp}}
\newcommand{\regcone}{\textit{alg.cone}}

\newcommand{\cut}{\ensuremath{\textit{cut}}\,}
\newcommand{\polyhedralcut}{\ensuremath{\overline{\textit{cut}}}\,}
\newcommand{\project}{\textit{project}}
\newcommand{\leadmon}{\textsc{Lm}}
\newcommand{\red}{\mathbf{red}}
\newcommand{\ThCR}{\mathbf{CR}}
\newcommand{\ThQ}{\mathbf{LRR}}
\newcommand{\ThZ}{\mathbf{LIRR}}
\newcommand{\Cn}[1]{\mathbf{C}\left(#1\right)}
\newcommand{\CnX}[2]{\mathbf{C}_{#2}\!\left(#1\right)}
\newcommand{\ConeModel}[1]{\mathfrak{M}(#1)}
\newcommand{\CnM}[1]{\mathcal{C}(#1)}
\newcommand{\algconeX}[2]{\textit{alg.cone}_{#2}(#1)}
\newcommand{\units}[1]{\mathcal{U}\!(#1)}
\newcommand{\linpoly}[1]{\mathbb{Q}[#1]^1}
\newcommand{\gb}[2]{\textit{Gr\"{o}bnerBasis}_{#2}(#1)}
\newcommand{\reducecone}[2]{\textit{reduce}_{#2}(#1)}
\newcommand{\polyhedralproject}[2]{\overline{\textit{project}}_{#2}(#1)}
\newcommand{\LRA}{\textbf{LRA}}
\newcommand{\NRA}{\textbf{NRA}}
\newcommand{\NIA}{\textbf{NIA}}

\newcommand{\Th}{\textit{Th}}

\begin{document}
\title{When Less Is More: Consequence-Finding in a Weak Theory of Arithmetic}

\author{Zachary Kincaid}
\affiliation{%
  \institution{Princeton University}
  \city{Princeton}
  \state{NJ}
  \country{United States}}
\email{zkincaid@cs.princeton.edu}

\author{Nicolas Koh} \affiliation{%
  \institution{Princeton University}
  \city{Princeton}
  \state{NJ}
  \country{United States}}
\email{ckoh@cs.princeton.edu}

\author{Shaowei Zhu} \affiliation{%
  \institution{Princeton University}
  \city{Princeton}
  \state{NJ}
  \country{United States}}
\email{shaoweiz@cs.princeton.edu}

\input{abstract}



\maketitle

\input{introduction}
\input{background}
\input{rational}

\input{integer}
\input{invgen}
\input{evaluation}
\input{related}

\bibliography{references}

\input{appendix}

\end{document}

%% file: prelude.tex
\usepackage{subcaption}
\usepackage{breqn}
\usepackage{tikz}
\usepackage{pgfplots}
\usetikzlibrary{calc,matrix,arrows,automata,positioning,chains}
\usepackage{enumitem}

\theoremstyle{definition}

\newtheorem{example}{Example}[section]
\newtheorem{lemma}{Lemma}

\newtheorem{theorem}{Theorem}

\newcommand{\set}[1]{\left\lbrace #1 \right\rbrace}       

\newcommand{\false}{\textit{false}}
\newcommand{\defeq}{\triangleq}

\definecolor{black}{RGB}{0,0,0}
\definecolor{green}{RGB}{25,160,80}
\definecolor{red}{RGB}{192,57,43}
\definecolor{blue}{RGB}{41,128,185}
\definecolor{orange}{RGB}{255,99,0}
\definecolor{gray}{RGB}{75,75,75}
\definecolor{lightgrey}{RGB}{240,240,240}
\definecolor{purple}{RGB}{155, 89, 182}
\definecolor{darkgrey}{RGB}{52, 73, 94}

\usepackage[linesnumbered,noend]{algorithm2e}
\SetKwInOut{Input}{Input}
\SetKwInOut{Output}{Output}
\SetKwRepeat{Do}{do}{while}

\SetCommentSty{mycommfont}
\SetKwProg{Fn}{Function}{}{}
\SetKwProg{Sub}{Subroutine}{ begin}{end}
\SetKwComment{Inv}{\rm\textbf{invariant:} }{}

\usepackage{listings}
\lstdefinestyle{base}{
  language=C,
  emptylines=1,
  breaklines=true,
  basicstyle=\color{black}\fontfamily{pzc}\selectfont\tt,
  commentstyle=\color{gray}\rm\itshape,
  keywordstyle=\rm\bfseries,
  identifierstyle=\rm\itshape,
  escapechar=|,
  morekeywords=[1]{then,do},
  morekeywords=[2]{assert},
  morekeywords=[3]{requires,ensures},
  keywordstyle	= [2]\color{red}\bf,
  keywordstyle	= [3]\bf\color{gray},
  numberstyle=\footnotesize\color{gray},
  moredelim=**[is][\bf\color{green}]{@!}{!@},
  moredelim=**[is][\bf\color{red}]{@?}{?@},
  literate=
    {<=}{{$\leq$}}1
    {>=}{{$\geq$}}1
    {!}{{$\neg$}}1
    {!=}{{$\neq$}}1
    {||}{{$\lor$}}1
    {&&}{{$\land$}}1
    {->}{{$\rightarrow$}}1
    {_1}{$_{1}$}2
    {_2}{$_{2}$}2
    {_3}{$_{3}$}2
}
\newcommand{\cinline}[1]{\mbox{\lstinline[style=base,mathescape]{#1}}}

\DeclareFontFamily{U}  {MnSymbolC}{}
\DeclareFontShape{U}{MnSymbolC}{m}{n}{
    <-6>  MnSymbolC5
   <6-7>  MnSymbolC6
   <7-8>  MnSymbolC7
   <8-9>  MnSymbolC8
   <9-10> MnSymbolC9
  <10-12> MnSymbolC10
  <12->   MnSymbolC12}{}
\DeclareSymbolFont{MnSyC}{U}{MnSymbolC}{m}{n}
\DeclareMathSymbol{\righthalfcup}{\mathrel}{MnSyC}{184}

\newenvironment{mexample}[1][]{%
   \begin{example}%
   \pushQED{\qed}%
 }{%
   \popQED%
 \end{example}%
}

\usepackage{mathpartir} 
\usepackage{booktabs}   
\usepackage{hyperref}   
\usepackage[english]{babel}
\addto\extrasenglish{
    
}
\newcommand{\zak}[1]{{\color{blue}Zak says: #1}}
\newcommand{\nic}[1]{{\color{red}Nic says: #1}}
\newcommand{\shaowei}[1]{{\color{green}Shaowei: #1}}
\renewcommand{\zak}[1]{}
\renewcommand{\nic}[1]{}
\renewcommand{\shaowei}[1]{}

\usepackage{tcolorbox}
\usepackage{mathtools}
\usepackage{graphics}
\usepackage{multirow}


%% file: abstract.tex
\begin{abstract}
  This paper presents a theory of non-linear integer/real arithmetic and
  algorithms for reasoning about this theory. The theory can be conceived as an
  extension of linear integer/real arithmetic with a weakly-axiomatized
  multiplication symbol, which retains many of the desirable algorithmic
  properties of linear arithmetic. In particular, we show that the
  \textit{conjunctive} fragment of the theory can be effectively manipulated
  (analogously to the usual operations on convex polyhedra, the conjunctive
  fragment of linear arithmetic). As a result, we can solve the following
  consequence-finding problem: \textit{given a formula $F$, find the strongest
    conjunctive formula that is entailed by $F$}. As an application of
  consequence-finding, we give a loop invariant generation algorithm that is
  monotone with respect to the theory and (in a sense) complete. Experiments show that the invariants
  generated from the consequences are effective for proving safety properties of
  programs that require non-linear reasoning.
\end{abstract}


%% file: introduction.tex
\section{Introduction} \label{sec:introduction}

Linear arithmetic is well-understood and widely used across a broad range of applications.
The conjunctive fragment of linear arithmetic,
convex polyhedra (and their intersection with lattices),
enjoys a number of good properties which enables applications beyond satisfiability-testing.
For instance, within the programming language community,
convex polyhedra
have been used in invariant generation \cite{POPL:CH1978,CAV:CSS2003},
termination analysis \cite{VMCAI:PR2004},
optimization \cite{IJCAR:ST2012, POPL:LAKGC2014, bjorner-nuZ3-2015},
and program transformation \cite{Feautrier-1996,lengauer-1993}.
Non-linear arithmetic, on the other hand, is not recursively axiomatizable.
If we restrict solutions to integers, even testing satisfiability of a single equation
is undecidable \cite{matiyasevich-1970, davis-putnam-robinson-1961}.

As a result, solvers for non-linear integer arithmetic rely on heuristic reasoning techniques \cite{SAT:FGMSTZ2007,CADE:BLNRR2009, CASC:KCA2016,TCL:BLROR2019,VMCAI:Jovanovic2017}.  The use of such heuristics poses a barrier to proving that clients of the solver satisfy desirable properties, such as completeness of ranking function synthesis, or monotonicity of an invariant generation scheme.   And while heuristics are often effective in practice, they can also be unpredictable.  For instance, \citet{HHLNPZZ:2014} reports ``we found Z3’s theory of nonlinear arithmetic to be slow and unstable; small code changes often
caused unpredictable verification failures,'' which prompted the authors to develop information-hiding techniques to avoid triggering non-linear heuristics.

This paper develops the theory of \textit{linear integer/real rings} ($\ThZ$)
-- commutative rings extended with an order relation and an ``integer''
predicate that obey certain axioms from the theory of \textit{linear}
integer/real arithmetic.  $\ThZ$ is designed to serve as a generalization of
linear integer/real arithmetic, in the sense that:
\begin{itemize}
\item The conjunctive fragment of $\ThZ$ can be manipulated effectively,
  analogously to (integral) convex polyhedra.
\item $\ThZ$ does not lose any of the reasoning power of linear integer/real
  arithmetic (in a sense made precise in Theorems~\ref{thm:lrr-lra-complete}
  and~\ref{thm:lirr-lira-complete}).
\end{itemize}
Moreover, $\ThZ$ is axiomatized by Horn clauses, which implies existence of
\textit{minimal models} \cite{JACM:VK1976} (there is no need to perform case
splits to find a model of a conjunctive formula), and ensures that the theory
is convex \cite{Tinelli2003} (and stably infinite), so $\ThZ$ can be combined
with other theories via the Nelson-Oppen protocol \cite{TOPLAS:NO1979}.



The key technical contribution of this paper is a suite of algorithms for manipulating \textit{algebraic cones}---a set of polynomials that can be represented as the sum of a polynomial ideal and a polyhedral cone.  The role of algebraic cones in $\ThZ$ is analogous to that of convex polyhedra in the theory of linear real arithmetic, in that they correspond to the (positive) conjunctive fragment of the language (i.e., a system of polynomial equalities and inequalities).  They
serve as the basis of our decision procedure for the problem of testing satisfiability of a ground formula modulo $\ThZ$, wherein they can be seen as a representation of a Herbrand model.  Furthermore, we show that algebraic cones can be effectively manipulated like convex polyhedra: membership and emptiness are decidable, and algebraic cones are closed under intersection, sum, and projection.

Algorithms for algebraic cones arise as a 
marriage of techniques between polynomial ideals (based on Gr\"{o}bner bases) 
and convex polyhedra. Such combinations have been investigated in prior work (e.g., \cite{SAS:BRZ2005, CSL:Tiwari2005, PACMPL:KCBR18}), in the context of
incomplete heuristics for reasoning about real or integer arithmetic.
This paper investigates the strength of this combination through the lens of
a first-order theory of arithmetic.
The critical finding is that these methods enable complete consequence-finding (modulo $\ThZ{}$).  We present an algorithm that, given a formula $F$, computes the set of all polynomials $p$ such that $F$ entails that $p$ is non-negative, modulo $\ThZ$ (analogous to computing the convex hull of $F$, modulo linear arithmetic).
Such consequence-finding has a wide range of applications in program analysis \cite{VMCAI:RSY2004}.  As a case study, we give one application to non-linear invariant generation,
and show that the technique has good practical performance on top of theoretical guarantees.




The paper is organized as follows.  Section~\ref{sec:background} presents
background on commutative algebra and polyhedral theory.  The theory of
linear/integer rings is presented in two steps.  Section~\ref{sec:rational}
presents the theory of \textit{linear real rings}, $\ThQ$, which is
essentially the theory of linear integer/real rings excluding the integer
predicate and its associated axioms.  The full theory $\ThZ$ is given in
Section~\ref{sec:integer}.  An invariant generation algorithm that
demonstrates the use of consequence-finding is given in
Section~\ref{sec:invgen}.  Section~\ref{sec:evaluation} evaluates the decision
procedure for $\ThZ$ and the invariant generation algorithm experimentally.
The $\ThZ$ decision procedure is not empirically competitive with
state-of-the-art heuristic solvers.  On the other hand, the experimental
results for the invariant generation procedure are positive, establishing the
value of consequence-finding modulo $\ThZ$.  Related work is discussed in
Section~\ref{sec:relwork}.  Statements that are made without proof are proved
in the Appendix.


%% file: background.tex
\section{Background} \label{sec:background}

\subsection{Commutative algebra}

This section recalls some basic facts about commutative algebra (see \cite{Book:CLO2015}).
Let $\sigma_{or}$ be the signature of ordered rings, consisting of a
binary addition (+) and multiplication $(\cdot)$ operators, the constants 0 and 1, equality, and a binary relation $\leq$.
For any set of symbols $X$, we use $\sigma_{or}(X)$ to denote the extension of $\sigma_{or}$ with the constant symbols $X$.

The \textbf{commutative ring axioms}, $\ThCR$, are as follows:
\[ \begin{array}{cr}
    \forall x,y,z\ldotp x+(y+z) = (x+y)+z \text{ and } \forall x,y,z\ldotp x \cdot (y \cdot z) = (x \cdot y) \cdot z & \text{(Associativity)}\\
    \forall x,y\ldotp x+y = y+x\text{ and  }\forall x,y\ldotp x\cdot y = y \cdot x & \text{(Commutativity)}\\
    \forall x\ldotp x+0 = x \text{ and } \forall x\ldotp x\cdot 1 = x & \text{(Identity)}\\
    \forall x,y,z\ldotp x\cdot(y + z) = (x \cdot y) + (x \cdot z) & \text{(Distributivity)}\\
    \forall x, \exists y\ldotp x + y = 0 & \text{(Additive inverse)}
\end{array}\]
A model of these axioms is called a \textbf{commutative ring}.
Examples of commutative rings include the integers $\mathbb{Z}$, the rationals $\mathbb{Q}$, and the reals $\mathbb{R}$.  For any commutative ring $R$ and finite set of variables $X$, let $R[X]$ denote the set of polynomials over $X$ with coefficients in $R$; this too forms a commutative ring.

\textit{Modules} are a generalization of linear spaces in which the scalars form a ring rather than a field.  If $R$ is a commutative ring, an \textbf{$R$-module} is a commutative group $\tuple{M, 0, +}$ equipped with a \textit{scalar multiplication} operation $\cdot : R \times M \rightarrow M$ satisfying the usual axioms of linear spaces ($a\cdot(m + n) = a\cdot m + a \cdot n$, $(a+b)\cdot m = a \cdot m + b \cdot m$, $(a\cdot b)\cdot m = a \cdot (b \cdot m)$, and $1 \cdot m = m$).  For instance, $R$ is itself an $R$-module where scalar multiplication is the usual ring multiplication; $R[X]$ is both an $R$-module and an $R[X]$-module.

Let $R$ be a commutative ring and $N$ be an $R$-module. For $L,M \subseteq N$ and $S \subseteq R$ we use $L+M$ to denote the set of sums of elements in $L$ and $M$, and $\product{S}{M}$ 
 to denote the set of weighted sums of elements of $M$ with coefficients in $S$:
\begin{align*}
  L + M &\defeq \set { \ell + m : \ell \in L, m \in M }\\
  \product{S}{M} &\defeq \set{ s_1m_1 + \dots + s_nm_n : n \in \mathbb{N}, s_1,\dots,s_n \in S, m_1,\dots,m_n \in M }
\end{align*}
For instance, if $V$ is a linear space over $\mathbb{Q}$ and $G \subseteq V$ then $\product{\mathbb{Q}}{G}$ is the \textit{span} of $G$---the smallest linear subspace of $V$ containing $G$.

Let $R$ be a commutative ring.
An \textbf{ideal} $I \subseteq R$ is a sub-module of $R$ (considered as an $R$-module);
i.e., a set that (1) contains zero, (2) is closed under addition, and
(3) is closed under multiplication by arbitrary elements of $R$.
An ideal $I$ defines a congruence relation $\equiv_I$, where $p \equiv_I q$
if and only if $p-q \in I$.
(The notation $p - q$ abbreviates
$p + (-q)$, where $-q$ is the unique additive inverse of $q$.)
One may think of an ideal as a set of elements that are congruent to zero with
respect to \textit{some} congruence relation, with the closure conditions of
ideals corresponding to the idea that the sum of two zero-elements is zero,
and the product of a zero-element with anything is again zero.
We use $R/I$ to denote the \textbf{quotient ring} in which the elements are
sets of the form $r + I$ for some $r \in R$ (that is, equivalence classes of $\equiv_I$),
and sum and product are defined as $(r+I)+(s+I) = (r+s) + I$ and
$(r + I)\cdot(s+I) = (r\cdot s) + I$.
Note that any subset $P \subseteq R$ generates an ideal
(the \textit{smallest} with respect to inclusion order containing $P$),
which is exactly $\product{R}{P}$.
When $R$ is clear from context,
we will write $\qideal{P}$ for the ideal $\product{R}{P}$ generated by
the set $P$, and $\qideal{p}$ for $\qideal{\set{p}}$.

Fix a set of variables $X$. We use $[X]$ to denote the set of monomials over $X$.
A \textbf{monomial ordering} $\preceq$ is a total order on
$[X]$ such that
(1) $1 \preceq m$ for all $m$,
and (2) for any $m \preceq n$ and any monomial $v$, we have $mv \preceq nv$.
(Assuming a fixed monomial ordering) the \textbf{leading monomial} $\leadmon(p)$ of a polynomial
$p = a_1 m_1 + \dotsi + a_n m_k \in \mathbb{Q}[X]$, $a_1, \ldots, a_n \neq 0$,
is the greatest monomial among $m_1, \dots, m_k$.
The leading monomial of the zero polynomial is undefined.


Fix a monomial ordering $\preceq$.  Let $p$ be a non-zero polynomial in $\mathbb{Q}[X]$.  Then $p$ can be written as $p = am + q$ where $m = \leadmon(p)$, $a$ is the coefficient of $m$ in $p$, and $q = p - am$; $p$ can be interpreted as a rewrite rule $m \rightarrow -\frac{1}{a}q$.  For instance, the polynomial $\frac{1}{2}x^2 -y$ can be interpreted as a rewrite rule $x^2 \rightarrow 2y$, and using this rule we may rewrite $x^3 + x^2 \rightarrow 2xy + x^2 \rightarrow 2xy + 2y$.
Applying a rewrite rule to a non-zero polynomial reduces it to $0$ or 
decreases the leading monomial.
A set of polynomials $G$ is a \textbf{Gr\"{o}bner basis} (with respect to $\preceq$) if its associated rewrite system is confluent.  
Assuming that $G$ is a Gr\"{o}bner basis, we use $\red_G : \mathbb{Q}[X] \rightarrow \mathbb{Q}[X]$ to denote the function that maps any polynomial to its normal form under the rewrite system $G$.  Equivalently, $\red_G(p)$ is the unique polynomial $q$ such that  (1) $p - q \in \qideal{G}$, and (2) no monomial in $q$ is divisible by $\leadmon(g)$ for any $g \in G$.
We note some important properties of $\red_G$:
\begin{itemize}
\item (Membership) For all $p \in \mathbb{Q}[X]$, $\red_G(p) = 0$ if and only if $p \in \qideal{G}$
\item (Linearity) If $a_1,\dots,a_n \in \mathbb{Q}$ and $p_1,\dots,p_n \in \mathbb{Q}[X]$, then $\reduce_G(\sum_{i=1}^n a_i p_i) = \sum_{i=1}^n a_i \reduce_G(p_i)$ 
\item (Ordering)
  For all $p \in \mathbb{Q}[X]$,
  $p = q_1g_1 + \dotsi + q_ng_n  + \reduce_G(p)$ for some
  $q_1,\dots,q_n \in \mathbb{Q}[X]$,
  $g_1, \ldots, g_n \in G$,
  and $\leadmon(q_ig_i) \preceq \leadmon(p)$ for all $i$. 
\end{itemize}

For any finite set of polynomials $P$ and monomial ordering $\preceq$,
we may compute a Gr\"{o}bner basis for the ideal generated by $P$
(e.g., using Buchberger's
algorithm~\cite{SIGSAM:Buchberger1976}).
We use $\gb{P}{\preceq}$ to denote this basis.

We make use of two monomial orderings in this paper. The degree reverse
lexicographic order compares monomials by their total degree first, and breaks
ties by taking the smallest monomial w.r.t. a reversed lexicographic order. This
ordering is used when we do not care about the exact order, since this often
leads to efficient Gr\"{o}bner basis computations. The elimination ordering is
used to project a set of variables out of an ideal (see
Section~\ref{sec:subsub-proj-alg-cones}).

\subsection{Polyhedral theory}

This section recalls some basic facts about polyhedral theory (see for example \cite{Book:Schrijver1999}).
In the following, we use \textbf{linear space} to refer to a linear space over the field $\mathbb{Q}$.  We use $\mathbb{Q}^n$ to denote the linear space of rational vectors of length $n$.
Note that for any set of variables $X$, $\mathbb{Q}[X]$ is an (infinite-dimensional) linear space.  We use $\linpoly{X}$ to denote the ($|X|+1$-dimensional) linear space of polynomials of degree at most one (i.e., polynomials of the form $a_1x_1 + \dots a_nx_n + b$, with $a_1,\dots,a_n,b \in \mathbb{Q}$ and $x_1,\dots,x_n \in X$).


Let $V$ be a linear space.  A set $C \subseteq V$ is called \textbf{convex} if for every $u,v \in C$, the line segment between $u$ and $v$ is contained in $C$ (i.e.,
for every $\lambda$ in the interval $[0,1]$, we have $\lambda u + (1-\lambda)v \in C$).  For a set $G \subseteq V$, we use $\conv{G}$ to denote the \textbf{convex hull} of $G$---the smallest convex set that contains $G$:
\[
 \conv{G} \defeq \set{  \lambda_1g_1 + \dots + \lambda_ng_n : \lambda_1,\dots,\lambda_n \in \Qplus, g_1, \dots, g_n \in G, \sum_{i=1}^n \lambda_i = 1}
\]
where $\Qplus$ denotes the set of non-negative rational numbers.
A set $C \subseteq V$ is called a (convex) \textbf{cone} if it contains zero and
is closed under addition and multiplication by non-negative scalars.
For a set $G \subseteq V$, the \textbf{conical hull} of $G$ is the smallest
convex cone containing $G$; it is precisely $\product{\Qplus}{G}$. 


Let $V$ be a linear space and $C \subseteq V$ be a cone.  We say that $v \in C$ is an \textbf{additive unit} if both $v$ and $-v$ belong to $C$, and denote the set of additive units as $\units{C}$ ($\units{C}$ is also known as the \textit{lineality space} of $C$---the largest linear space contained in $C$).  We say that a $C$ is \textbf{salient} if $\units{C} = 0$.  We say that $C$ is \textbf{finitely generated} (or polyhedral) if $C = \cone{G}$ for some finite set $G$.

Let $X$ be a finite set of variables.
The set of functions $\mathbb{Q}^X$ mapping variables to rationals is a linear space.
A \textbf{polyhedron} in $\mathbb{Q}^X$
is a subset  of the form  $P = \set{ x \in \mathbb{Q}^X: p_1(x) \geq 0, \dots, p_n(x) \geq 0 }$ for some linear polynomials $p_1,\dots,p_n$.  We say that inequality $p(x) \geq 0$ is \textbf{valid} for a polyhedron $P$ if it is satisfied by every point in $P$.  The set of valid inequalities of a polyhedron $P = \set{ x \in \mathbb{Q}^X : p_1(x) \geq 0, \dots, p_n(x) \geq 0 }$ forms a cone; by Farkas' lemma \cite[Cor 7.ld]{Book:Schrijver1999} this cone is precisely $\cone{p_1,\dots,p_n,1}$.  The \textbf{integer hull} $\inthull{P}$ of a polyhedron $P$ is defined to be
$\inthull{P} \defeq \conv{P \cap \mathbb{Z}^X}$, the convex hull of the integral points of $P$.  The integer hull of a polyhedron $P$ is itself a polyhedron, and its constraints can be computed from the constraints of $P$ by the \textit{iterated Gomory-Chv\'{a}tal closure} process \cite[Ch. 23]{Book:Schrijver1999}.
A \textit{dual view} of the integer hull of a polyhedron is the 
\textit{cutting plane closure} of its valid inequalities \cite{Chvatal1973}.
For an example illustrating the intuition behind cutting planes, suppose that we know that $2x -1 \geq 0$, and that $x$ takes integer values---then we must have $x - \frac{1}{2} \geq 0$, and thus $\floor{x - \frac{1}{2}} = x + \floor{-\frac{1}{2}} = x - 1 \geq 0$.  That is, we may ``shift'' the halfspace  to eliminate some real solutions while keeping all integer points.  A set $C \subseteq \linpoly{X}$ is \textit{closed under cutting planes} iff for any $a, b \in \mathbb{Z}$ with $a > 0$ and any $p \in \mathbb{Z}[X]$ such that $ap + b \in C$, we have
$p + \floor{b/a} \in C$.  We use $\cp(C)$ to denote the cutting plane closure of a cone $C$---the smallest cone that contains $C$ and which is closed under cutting planes. Assuming that $C$ is finitely generated, $\cp(C)$ can be computed by any algorithm for computing the integer hull of convex polyhedra.

Let $V$ be a linear space. We say that a subset $L \subseteq V$ is a
\textbf{point lattice}
if there exists some $v_1,\dots,v_n \in V$ such that
$L = \product{\mathbb{Z}}{v_1,\dots,v_n}$.
We call a set of generators $\set{v_1,\dots,v_n}$ for a point lattice $L$ a
\textbf{basis} for $L$ if it is linearly independent.
A basis for a point lattice can be computed from its generators in polynomial time
\cite[Ch. 4]{Book:Schrijver1999}.


%% file: rational.tex
\section{Linear Real Rings} \label{sec:rational}
In this section, we develop the theory of  \textit{linear real rings}, $\ThQ$.  Linear real rings are a common extension of the theory of commutative rings and the positive fragment of the theory of linear real arithmetic. 
Section \ref{sec:weak-theory-of-rational-arith} presents the axioms of the theory and introduces \textit{regular} and \textit{algebraic} cones.  Regular cones correspond to models of the theory, and algebraic cones correspond to effective structures; cones that are both regular and algebraic correspond to a class of effective models of $\ThQ$.  Section~\ref{sec:sat-mod-Q} presents a decision procedure for satisfiability of ground formulas modulo $\ThQ$.
Section~\ref{sec:conseq-finding-mod-Q} gives a procedure that discovers all implied inequalities of a ground formula modulo $\ThQ$, represented as a (regular) algebraic cone.

\subsection{Linear real rings} \label{sec:weak-theory-of-rational-arith}

Define $\ThQ$ to be the $\sigma_{or}$-theory axiomatized by the theory of commutative rings
$\ThCR$ as well as the following axioms derived from \LRA{}:
\[\begin{array}{cr}
\forall x\ldotp x \leq x & \text{(Reflexivity)}\\
\forall x,y,z\ldotp x \leq y \land y \leq z \Rightarrow x \leq z & \text{(Transitivity)}\\
\forall x,y,z\ldotp x \leq y \land y \leq x \Rightarrow x = y & \text{(Antisymmetry)}\\
\forall x,y,z\ldotp (x \leq y \Rightarrow x + z \leq y + z) & \text{(Compatibility)}\\
0 \leq 1 \land 0 \neq 1 & \text{(Non-triviality)}\\
\text{for all } n \in \mathbb{N}^{\geq 1}, \exists x\ldotp \underbrace{x + \dotsi + x}_{n \text{ times}} = 1 & \text{(Divisibility)}\\
\text{for all } n \in \mathbb{N}, \forall x\ldotp 0 \leq \underbrace{x + \dotsi + x}_{n \text{ times}} \Rightarrow 0 \leq x & \text{(Perforation-freeness)}
\end{array}\]

Note that divisibility and perforation-freeness are axiom schemata, with one axiom for each natural number ($\exists x\ldotp x = 1$, $\exists x \ldotp x + x = 1$, $\exists x \ldotp x + x + x = 1$, and so on).
Equivalently, $\ThQ$ is the theory consisting of all $\sigma_{or}$-sentences that hold in
all structures $\mathfrak{A}$ where $(U^{\mathfrak{A}},0^{\mathfrak{A}},1^{\mathfrak{A}},+^{\mathfrak{A}},\cdot^{\mathfrak{A}})$ forms a commutative ring and where
($U^{\mathfrak{A}},0^{\mathfrak{A}},+^{\mathfrak{A}}, \leq^{\mathfrak{A}}$) forms an unperforated partially ordered divisible group.  We regard the real numbers $\mathbb{R}$ as the ``standard'' model of $\ThQ$; other models include the rationals, the complex numbers (where complex numbers with the same imaginary part are ordered by their real part) and, as we shall see shortly, more exotic interpretations.  The signature of ordered rings should be regarded as a ``minimal'' signature for $\ThQ$, but models of $\ThQ$ can be lifted to larger languages.  In the following, we will make use of an extended language that includes
atomic formulas $p \leq q$ and $p = q$ where $p$ and $q$ have rational coefficients; any such atom can be translated into $\sigma_{or}(X)$ by scaling and re-arranging terms (for example, the formula $0 \leq \frac{1}{2}x - y$ can be translated to the $\sigma_{or}(\set{x,y})$ formula $(1+1) \cdot y \leq x$).

The axioms that $\ThQ$ adds onto commutative rings are a subset of those for linear real arithmetic.  Naturally we might ask if the subset is ``enough.''  Notably, totality of the order (an axiom of linear real arithmetic) is independent of $\ThQ$ (e.g., $\leq$ is non-total for the complex numbers).
Nevertheless, $\ThQ$ is at least as strong as \LRA{}, in the sense that satisfiability modulo \LRA{} can be reduced to satisfiability modulo $\ThQ$ (noting that every formula in the language of \LRA{} is equivalent to a negation-free formula modulo \LRA{}).
\begin{theorem} \label{thm:lrr-lra-complete}
  Let $F$ be a ground negation-free formula in the language of \LRA{}.
  Then $F$ is satisfiable modulo \LRA{} iff it is satisfiable modulo $\ThQ$.
\end{theorem}
\begin{proof}
  The $\Rightarrow$ direction is trivial, since the reals are a model of
  $\ThQ$. For the $\Leftarrow$ direction, we show the contrapositive: suppose
  that $F$ is unsatisfiable modulo \LRA{}, and show that it is unsatisfiable modulo
  $\ThQ$. Without loss of generality, we may suppose that $F$ is conjunction of
  inequalities, which can be written in the form $A\vec{x} \leq \vec{b}$. Since
  this system is unsatisfiable modulo \LRA{}, then by Farkas' lemma there is some
  $\vec{y} \geq 0$ such that $\transpose{\vec{y}}A = 0$ and
  $\transpose{\vec{y}}\vec{b} < 0$.
  Without loss of generality, we may suppose
  that $\vec{y}$ is integral.
  By compatibility, commutativity, associativity, and transitivity, we have
  $F \models_{\ThQ} \transpose{\vec{y}}A\vec{x} \leq \transpose{\vec{y}}\vec{b}$.
  (Suppose $F \models c_1 \leq d_1$ and $F \models c_2 \leq d_2$.
  From the first, $F \models c_1 + c_2 \leq d_1 + c_2$ by compatibility;
  from the second, $F \models d_1 + c_2 \leq d_1 + d_2$
  by compatibility and commutativity.
  By transitivity, $F \models c_1 + c_2 \leq d_1 + d_2$.
  So we can take non-negative integer multiples of inequalities and add them up.)

  Since $\leq$ is unperforated, we have $F \models_{\ThQ} 0 \leq -1$, and by
  non-triviality and antisymmetry we have $F \models_{\ThQ} \false$; i.e., $F$
  is unsatisfiable modulo $\ThQ$.

  Note that the proposition holds even if we consider formulas with disequalities
  (and thus the case of strict inequalities, treating $p < q$ as an abbreviation for
  $p \leq q \land p \neq q$).
  In this case, we may suppose without loss of generality that $F$ is conjunction of
  inequalities and disequalities, which can be written in the form
  $A\vec{x} \leq \vec{b} \land \bigwedge_{i=1}^n \transpose{\vec{c}}_i\vec{x} \neq d_i$
  (with at least one disequality, or else we fall into the case above).
  Since $F$ is unsatisfiable, we must have
  $A\vec{x} \leq \vec{b} \models_{\LRA{}} \bigvee_{i=1}^n \transpose{\vec{c}}_i\vec{x} = d_i$.
  Since $\LRA{}$ is a convex theory, we must have
  $A\vec{x} \leq \vec{b} \models_{\LRA{}} \transpose{\vec{c}}_i\vec{x} = d_i$ for some $i$.  We may then argue as above that
  $A\vec{x} \leq \vec{b} \models_{\ThQ} \transpose{\vec{c}}_i\vec{x} \leq d_i$
  and $A\vec{x} \leq \vec{b} \models_{\ThQ} d_i \leq \transpose{\vec{c}}_i\vec{x}$, and so by antisymmetry $A\vec{x} \leq \vec{b} \models_{\ThQ} \transpose{\vec{c}}_i\vec{x} = d_i$, and thus $F$ is unsatisfiable modulo $\ThQ$.
\end{proof}

In the remainder of this section, we develop a model theory of $\ThQ$, based on \textit{regular cones}.
For any set of variables $X$, we say that a set
$C \subseteq \mathbb{Q}[X]$ is a \textbf{regular cone} if it is a cone (closed
under addition and multiplication by nonnegative rationals), $1 \in C$, and
$\units{C}$ forms an ideal in $\mathbb{Q}[X]$. We say that $C$ is
\textbf{consistent} if $C \neq \mathbb{Q}[X]$; in the case that $C$ is regular, $C$ is consistent iff $-1 \notin C$.

Let $X$ be a set of symbols, and let $\mathfrak{A}$ be a $\sigma_{or}(X)$-structure
satisfying the axioms of $\ThQ$.
  Define $\CnM{\mathfrak{A}} \defeq \set{ p \in \mathbb{Q}[X]: \mathfrak{A} \models 0 \leq p}$ to be
its nonnegative consequences. Naturally, $\mathcal{C}(\mathfrak{A})$ forms a consistent regular cone.
We now show that, conversely, any consistent regular cone can be associated with a model of $\ThQ$.


Let $X$ be a set of variables, let $C \subseteq \mathbb{Q}[X]$ be a regular cone.  Define a $\sigma_{or}(X)$-structure $\ConeModel{C}$ where
\begin{itemize}
\item The universe and function symbols $0,1+,\cdot$ are interpreted as in the quotient ring $R \defeq \mathbb{Q}[X]/I$, where $I$ is the ideal of additive units $\units{C}$.  Thus, elements of the universe are sets of polynomials with rational coefficients of the form $p + I$, where $p \in \mathbb{Q}[X]$.
\item Each constant symbol $x \in X$ is interpreted as $x + I$.
\item $\leq$ is interpreted as the relation $\{ (p + I, q + I) : q - p \in C \}$.
\end{itemize}
Observe that $\ConeModel{C} \models 0 \leq q$
  iff $q \in C$, and $\ConeModel{C} \models 0 = q$ iff $q \in \units{C}$.

\begin{lemma} \label{lem:model-construction}
  Let $X$ be a set of variables, and let $C \subseteq \mathbb{Q}[X]$
  be a consistent regular cone.  Then $\ConeModel{C}$ is a model of
  $\ThQ$.
\end{lemma}
\begin{proof}
  Clearly $R = \mathbb{Q}[X]/\units{C}$ is a commutative ring that satisfies divisibility. Since $C$ is a cone, it follows that $\leq$ is reflexive, transitive, compatible with addition, and perforation-free; furthermore since $C$ is regular we have that $\leq$ is antisymmetric.  Non-triviality follows from the fact that $C$ is consistent.
\end{proof}


While regular cones give us a ``standard form'' in which to represent models of $\ThQ$, they cannot be manipulated effectively.
For this purpose, we introduce \textit{algebraic cones},
which are (not necessarily regular) cones that admit a finite representation.


\zak{Move to SAT section?} \shaowei{I think the current flow is good. I don't see why we would want
to move this and the utility functions to SAT.}
We say that a cone $C \subseteq \mathbb{Q}[X]$ is \textbf{algebraic} if there is an ideal $I$ and a finitely-generated cone $D$ such that $C = I + D$.  An algebraic cone can be represented as a
pair $\tuple{Z,P}$ ($Z, P \subseteq \mathbb{Q}[X]$) where  $Z = \set{ z_1, \dots, z_m }$ (``zeros'') is a basis for an ideal and $P = \set{ p_1, \dots, p_n}$ (``positives'') is a basis for a cone; the algebraic cone represented by $\tuple{Z,P}$ is denoted by
\[
\algconeX{Z,P}{X} \defeq \qideal{Z} + \cone{P}
=\set{  \sum_{j=1}^m q_iz_i + \sum_{i=1}^n \lambda_ip_i : q_1,\dots,q_m \in \mathbb{Q}[X], \lambda_1,\dots,\lambda_n \in \Qplus }
\]
We will omit the $X$ subscript when it is clear from context.  Say that the pair $\tuple{Z,P}$ is \textbf{reduced} (with respect to a monomial ordering $\preceq$) if
\begin{enumerate}
\item $Z$ is a Gr\"{o}bner basis for $\qideal{Z}$ (with respect to $\preceq$)
\item Each $p_i \in P$ is reduced with respect to $Z$ (i.e., $\red_Z(p_i) = p_i$ for all $i$).
\end{enumerate}

The following shows that the problem of checking membership in an algebraic cone can be reduced to checking membership in a finitely-generated cone (which can be checked in polytime using linear programming).  It comes in two parts: (1) checking membership assuming a \textit{reduced} representation (Lemma~\ref{lem:membership}) and (2) computing a reduced representation (Lemma~\ref{lem:ordering}).
\begin{lemma}[Membership]
If $\tuple{Z,P}$ is reduced, then for any polynomial $p \in \mathbb{Q}[X]$, we have $p \in \regcone(Z,P)$ iff $\reduce_Z(p) \in \cone{P}$.
\label{lem:membership}
\end{lemma}
\begin{proof}
  Let $q \in \qideal{Z}$ be such that $p = \reduce_Z(p) + q$.
  Then we have  $\reduce_Z(p) \in \cone{P}$ iff $p = \reduce_Z(p) + q \in \cone{P} + \qideal{Z} = \regcone(Z,P)$.
\end{proof}

For any set of variables $X$,
finite sets of polynomials $Z,P \subseteq \mathbb{Q}[X]$,
and monomial ordering $\preceq$, define
$\reducecone{Z,P}{\preceq} \defeq \tuple{G, \set{ \reduce_G(p) : p \in P, \reduce_G(p) \neq 0 }} $
where $G = \gb{Z}{\preceq}$ is a Gr\"{o}bner basis for $\qideal{Z}$ with respect to the order $\preceq$.

\begin{lemma}[Ordering] \label{lem:ordering}
  Let $X$ be a set of variables, $Z,P \subseteq \mathbb{Q}[X]$ be finite sets of polynomials, and $\preceq$ be a monomial ordering.
  Then $\reducecone{Z,P}{\preceq}$ is reduced with respect to $\preceq$ and
  $ \regcone(Z,P) = \regcone(\reducecone{Z,P}{\preceq})$.
\end{lemma}
\begin{proof}
  Let $G = \gb{Z}{\preceq}$ and $P' = \set{ \reduce_G(p) : p \in P, \reduce_G(p) \neq 0 }$.
  Since $\red_G$ is idempotent, $\reducecone{G, P'}{\preceq}$ is reduced w.r.t. $\preceq$.
  Clearly, $\qideal{Z} = \qideal{G}$.
  Since algebraic cones are closed under addition and multiplication by non-negative rationals,
  it is sufficient to prove that
  $P \subseteq \qideal{G} + \cone{P'}$ and
  $P' \subseteq \qideal{Z} + \cone{P}$.
  \begin{itemize}
  \item $P \subseteq \qideal{G} + \cone{P'}$:
    Suppose $p \in P$.
    Then $p = z + \reduce_{G}(p)$ for some $z \in \qideal{G}$
    (since $p - \red_{G}(p) \in \qideal{G}$ by the definition of reduction).
    We have $\red_{G}(p) \in \cone{P'}$ (since it either belongs to $P'$ or it is zero),
    and so $p = z + \reduce_{G}(p) \in \qideal{G} + \cone{P'}$.
  \item $P' \subseteq \qideal{Z} + \cone{P}$:
    Suppose $p' \in P$.
    Then $p' = \reduce_{G}(p)$ for some $p \in P$,
    and so $p' = z + p$ for some $z \in \qideal{G} = \qideal{Z}$;
    thus $p' \in \qideal{Z} + \cone{P}$. \qedhere
  \end{itemize}
\end{proof}

The above procedure also allows us to do model checking w.r.t. models
associated to algebraic cones. Given a ground formula $F$, we have a decision
procedure to check if $\ConeModel{\regcone(Z,P)} \models F$.




\subsection{Satisfiability modulo $\ThQ$} \label{sec:sat-mod-Q}

This section presents a decision procedure for testing satisfiability of
ground $\sigma_{or}(X)$ formulas modulo the theory $\ThQ$. As usual, it is sufficient to develop a
\textit{theory solver}, which can test satisfiability of the conjunctive
fragment; formulas with disjunctions can be accommodated using DPLL($\mathcal{T}$)
\cite{CAV:GHNOT2004}.

Without loss of generality, a conjunctive formula $F$ can be written in the form
\[ F = \left(\bigwedge_{p \in P} 0 \leq p \right) \land \left( \bigwedge_{q \in Q} \lnot (0 \leq q) \right) \wedge \left( \bigwedge_{r \in R} \lnot(0 = r) \right) \]
where $P$, $Q$, and $R$ are finite sets of polynomials (noting the equivalences $x \leq y \equiv 0 \leq y - x$ and
$x = y \equiv 0 \leq x - y \land 0 \leq y - x$).  In the following, we will first show that it is possible to compute a finite representation of the \textit{least regular cone} $C$ that contains all of the non-negative polynomials $P$ (Theorem~\ref{thm:saturation}), and then show that $F$ is satisfiable if and only if $C$ is consistent and $\ConeModel{C} \models F$ (Theorem~\ref{thm:sat-correctness}).  Since $C$ is algebraic and computable from $F$, and checking that $C$ is consistent and that $\ConeModel{C} \models F$ is decidable, this yields a sound and complete procedure for checking satisfiability of ground $\sigma_{or}(X)$-formulas modulo $\ThQ$.

\begin{algorithm}[t]
\SetKwFunction{FSaturate}{saturate}

\Fn{\FSaturate{$Q$}}{
\Input{Set of polynomials $Q$}
\Output{Reduced pair $\tuple{Z,P}$ such that $\regcone(Z,P)$ is the least regular cone that contains $Q$}
$\tuple{Z,P} \gets \tuple{\emptyset, Q \cup \set{1}}$\;
\While{there is some non-zero $t \in \units{\cone{P}}$}{
  \tcc{Sampling from $\units{\cone{P}}$ can be implemented by (e.g.) linear programming}
  $\tuple{Z,P} \gets \reducecone{Z \cup \set{t}, P}{\preceq}$\;
}
\Return{$\tuple{Z,P}$}
}

\caption{Saturation \label{alg:saturation}}
\end{algorithm}

\begin{theorem} \label{thm:saturation}
  For any finite set of polynomials $Q \subseteq \mathbb{Q}[X]$, \textsf{saturate}($Q$) (\autoref{alg:saturation}) returns the least regular cone that contains $Q$.
  \label{thm:model-construction-alg-correct}
\end{theorem}
\begin{proof}

  Let $Z_{i}, P_{i}$, and $t_{i}$ denote the values of $Z$, $P$, and $t$ after $i$ iterations of the loop in \autoref{alg:saturation}.
  We first observe an invariant of the
  algorithm: by the ordering lemma (Lemma~\ref{lem:ordering}), we have that for all $i$, we have
  $\regcone(Z_{i+1}, P_{i+1}) = \regcone(\reducecone{Z_{i} \cup \set{t}, P_i}{\preceq}) = \regcone(Z_{i}, P_{i}) + \qideal{t_{i}}$.  This follows because at each iteration $i$, $\tuple{Z_{i+1},P_{i+1}}$ is computed by reducing the pair
  $\tuple{Z_{i} \cup \set{t}, P_i}$.

  We first prove that the algorithm terminates. For iteration $i+1$, we have
  $\qideal{Z_{i+1}} = \qideal{Z_{i}} + \qideal{t_{i}}$. Since $t_{i}$ is a conic combination of polynomials in $P_{i}$ which are 
  reduced with respect to $Z_{i}$, we have $\reduce_{Z_{i}}(t_{i}) = t_{i} \neq 0$ and so $t_{i} \notin \qideal{Z_{i}}$. Hence we have
  $\qideal{Z_{i+1}} \supsetneqq \qideal{Z_{i}}$.
  For a contradiction, suppose that the algorithm does not terminate.  Then we have
  an infinite strictly ascending chain of ideals $\qideal{Z_0} \subsetneqq \qideal{Z_1} \subsetneqq \dots$ in $\mathbb{Q}[X]$,
  which contradicts the fact that $\mathbb{Q}[X]$ is a Noetherian ring \cite[Ch. 2 \S 5]{Book:CLO2015}.

  We now show that \textsf{saturate}($Q$) computes the least regular cone that
  contains $Q$.
  Since $Q \subseteq P_0$, and we have
  $\regcone(Z_{i}, P_{i}) \subseteq \regcone(Z_{i+1}, P_{i+1})$ for all $i$, we have that $\regcone(\textsf{saturate}(Q))$ must contain $Q$; it remains to show that it is the \textit{least} such regular cone.  Suppose that there is another
 regular cone $C$ such that $C \supseteq Q$.  We show that for all iterations $i$,
  $\regcone(Z_{i}, P_{i}) \subseteq C$ by induction on $i$. Initially this is true since $P_0 = Q \cup \set{1}$,
  and $C$ contains both $Q$ (by assumption) and $1$ (since $C$ is regular).  For the inductive step, we suppose that $\regcone(Z_{i}, P_{i}) \subseteq C$ and prove that
  $\regcone(Z_{i+1}, P_{i+1}) =  \regcone(Z_{i}, P_{i}) + \qideal{t_i} \subseteq C$.
  Since $C$ is closed under addition and $\regcone(Z_{i}, P_{i}) \subseteq C$ by the inductive hypothesis, it is sufficient to show that 
  $\qideal{t_i} \subseteq C$.
  Since $t_i \in \units{\cone{P_{i}}}$ and $\cone{P_{i}} \subseteq C$, we must have
  $t_i \in \units{C}$.
  Since $C$ is regular, $\units{C}$ is an ideal, and so
 $\qideal{t_i} \subseteq \units{C} \subseteq C$.
\end{proof}

\begin{mexample}
  Table~\ref{tbl:saturation} illustrates Algorithm~\ref{alg:saturation} on the set of polynomials \[Q = \set{ x^2-xy, xy - x^2, x^2y - z, w - xy^2, z-w, w^3}\ .\]
  Each row $i$ gives the set of zero polynomials $Z$ and set of positive polynomials $P$   at the beginning of iteration $i$, along with the selected additive unit $t$.  Intuitively, each round selects an additive unit from the positives, removes it from the positives and adds it to the zeros.  The algorithm terminates at iteration 3: the positive cone is salient, and so there is no additive unit to select.
  \end{mexample}

  \begin{table}
\caption{Execution of Algorithm~\ref{alg:saturation} on input $Q = \set{x^2-xy,xy-x^2,x^2y-z,w-xy^2,z-w,w^3}$.} \label{tbl:saturation}
\begin{tabular}{cccc}
\toprule
Iteration & $Z$ & $P$ & Additive unit\\
\midrule
0 & $\emptyset$ & $\set{1} \cup Q $ & $x^2-xy$\\
1 & $\set{xy - x^2}$ & $\set{1, x^3 - z, w - x^3, z-w, w^3}$ & $x^3-z$\\
2 & $\set{xy-x^2,x^3-z,yz-xz}$  & $\set{1,w-z,z-w,w^3}$ & $w-z$\\
3 & $\set{xy-x^2,x^3-z,yz-xz,w-z}$ & $\set{1,z^3}$ & --\\
\bottomrule
\end{tabular}
  \end{table}

The following theorem is the basis of our decision procedure for $\ThQ$: it shows that to test satisfiability of $F$, we
only need to check whether the least cone $C$ that agrees with all 
positive atoms of $F$ is consistent and other atoms do not contradict the
consequences of $C$. 
It does so by observing that $C$ is a ``minimal'' model of $F$
in that any model of $F$ must entail all polynomials in $C$ to be nonnegative.

\begin{theorem} \label{thm:sat-correctness}
  Let $F$ be the ground conjunctive formula
  \[
  F = \left(\bigwedge_{p \in P} 0 \leq p \right)  \land \left( \bigwedge_{q \in Q} \lnot (0 \leq q) \right) \wedge \left( \bigwedge_{r \in R} \lnot(0 = r) \right).
  \]
  Let $C$ be the least regular cone that contains $P$.  Then $F$ is satisfiable iff $C$ is consistent and $\ConeModel{C} \models F$.
  \label{thm:min-model-for-conjunction}
\end{theorem}
\begin{proof}
  If $C$ is consistent and $\ConeModel{C} \models F$, then $F$ is certainly
  satisfiable modulo $\ThQ$. We thus only need to prove the other
  direction.

  Suppose that $F$ is satisfiable modulo $\ThQ$.  Then there is some model $\mathfrak{A}$ of $\ThQ$ with $\mathfrak{A} \models F$.  We have that
  $C \subseteq \CnM{\mathfrak{A}}$, since $\CnM{\mathfrak{A}}$ is a regular cone that contains $P$ (since $\mathfrak{A} \models F$), and $C$ is the least such cone. It follows that $C$ is consistent, since if $-1 \in C$ then $-1 \in \CnM{\mathfrak{A}}$, which is not possible because $\CnM{\mathfrak{A}}$ is consistent.
  It remains to show that
  $\ConeModel{C} \models F$. Clearly $\ConeModel{C} \models 0 \leq p$ for
  all $p \in P$, since $P \subseteq C$.  For any $q \in Q$ we have
  $\mathfrak{A} \models \neg( 0 \leq q )$, since $\mathfrak{A} \models F$.
  Thus $q \notin \CnM{\mathfrak{A}}$ and furthermore $q \notin C$
  using the claim. Hence $\ConeModel{C}\models \neg (0 \leq q)$. Similarly,
  we have that $\ConeModel{C} \models \neg (0 = r)$ for all $r \in R$. Combining the
  above, we have $\ConeModel{C} \models F$.
\end{proof}

\paragraph{Decision procedure for $\ThQ$}
Summarizing, we have the following decision procedure for satisfiability of
conjunctive $\sigma_{or}(X)$-formulas modulo $\ThQ$: given a formula
\[F = 
  \left(\bigwedge_{p \in P} 0 \leq p \right) \land \left( \bigwedge_{q \in Q}
  \lnot (0 \leq q) \right) \wedge \left( \bigwedge_{r \in R} \lnot(0 = r)
  \right)\] First compute a representation $\tuple{Z',P'}$
  of the least regular cone containing $P$ using
  Algorithm~\ref{alg:saturation}.  If $\red_{Z'}(1) = 0$, then $F$
  is unsatisfiable ($\regcone(Z',P')$ is inconsistent).
  Otherwise, check whether $\ConeModel{\regcone(Z',P')}
  \models F$ by testing whether there is some $q \in Q$ with
  $\red_{Z'}(q) \in \product{\Qplus}{P'}$
  (Lemma~\ref{lem:membership}), or some $r \in R$ with $\red_{Z'}(r)
  = 0$; if such a $q$ or $r$ exists, then $F$ is unsatisfiable
  (Theorem~\ref{thm:sat-correctness}), otherwise,
  $\ConeModel{\regcone(Z',P')}$ satisfies $F$.


\subsection{Consequence-finding modulo $\ThQ$} \label{sec:conseq-finding-mod-Q}
In this subsection, we give an algorithm that computes all polynomial
consequences of a formula modulo $\ThQ$.
Let $X$ be a set of symbols and let $F$ be a $\sigma_{or}(X)$ formula.
Define the \textbf{consequence cone} $\CnX{F}{X}$ of $F$ as follows:
\[
  \CnX{F}{X} \defeq \set{ p \in \mathbb{Q}[X] : F \models_{\ThQ} 0 \leq p }
\]
It is easy to verify that for any
formula $F$, $\CnX{F}{X}$ is a regular cone.  We will omit the $X$ subscript when it can be understood from the context.

A simple ``eager'' strategy for
computing consequence cones operates as follows.
Suppose that $F$ is a formula with free variables $Y$, and that $X \subseteq Y$.
$\CnX{F}{X}$ operates as follows. First, we place $F$ in  disjunctive normal form; i.e., we compute a
formula that is equivalent to $F$ and takes the form $\bigvee_{i=1}^n G_i$ where
each $G_i$ is a conjunctive formula. Observe that
\[ \CnX{F}{X}
= \CnX{\bigvee_{i=1}^n G_i}{X}
= \bigcap_{i = 1}^n \CnX{G_i}{X}
=\bigcap_{i = 1}^n (\CnX{G_i}{Y} \cap \mathbb{Q}[X]) \]
since a disjunctive formula entails that a polynomial is non-negative exactly when each disjunct does.  Thus, the problem of computing $\CnX{F}{X}$ can be reduced to three sub-problems: (1) computing the consequence cone of a \textit{conjunctive} formula, (2) projection of a cone onto a subset of variables (i.e., the intersection of a regular cone with $\mathbb{Q}[X]$), and the
intersection of regular cones.
In the following, we show how to solve each sub-problem.
Before we do, we remark that the above strategy
has a ``lazy'' variant that avoids (explicit) DNF computation,
Algorithm~\ref{alg:consequence}. The lazy variant operates by iteratively
selecting a cube from the DNF of $F$, computing its consequence cone
(Lemma~\ref{lem:conseq-for-conjunction}), and adding blocking clauses so that
the same cube would not be selected again in future iterations.



\begin{algorithm}

  \SetKwFunction{FConseq}{consequence}
  \SetKwFunction{FModel}{get-model}
  \SetKwFunction{FIntersect}{intersect}
  \SetKwFunction{FProj}{project}

  \Fn{\FConseq{$F,X$}}{
\Input{$\sigma_{or}(Y)$ formula $F$ and set of variables $X \subseteq Y$ }
\Output{Reduced pair $\tuple{Z,P}$ with $\algconeX{Z, P}{X} = \CnX{F}{X}$.}
$G \gets F$\;
$\tuple{Z,P} \gets \tuple{\set{1},\set{0}}$
\tcc*{Invariant: $\CnX{F}{X} \subseteq \regcone(Z,P)$}
\While{$G$ is satisfiable}{
\tcc{$\regcone(Z',P') = \CnX{G_i}{Y}$ for some cube $G_i$ of the DNF of $F$ (Lemma~\ref{lem:conseq-for-conjunction})}
$\tuple{Z', P'} \gets \FModel(G)$\; 
  $\tuple{Z',P'} \gets \FProj(Z',P', X)$ \;
  $\tuple{Z,P} \gets \FIntersect(Z,P,Z',P')$\;
  \tcc{Block any model $\mathfrak{A}$ with $\CnX{\mathfrak{A}}{X} \subseteq \regcone(Z,P)$}
  $G \gets G \land \lnot \left(\left(\bigwedge_{z \in Z} z = 0\right) \land \left(\bigwedge_{p \in P} 0 \leq p\right)\right)$
}
\Return{$\tuple{Z,P}$}
}
\caption{Lazy consequence finding \label{alg:consequence}}
\end{algorithm}

In the following, we first show how to solve the three sub-problems
that are used in consequence finding, before proving the correctness
of the procedure.

\subsubsection{Non-negative cones of conjunctive formulas}

First, we show how to compute the cone of non-negative polynomials of a \textit{conjunctive} formula.  In fact, we show that we have already shown how to compute this set of consequences: it is the least regular cone that contains the non-negative atoms of the formula, as computed by the \textit{saturate} procedure:
\begin{lemma} \label{lem:consequence-cone}
  Let $Y$ be a set of symbols and let
  \[
  F = \left(\bigwedge_{p \in P} 0 \leq p \right)  \land \left( \bigwedge_{q \in Q} \lnot (0 \leq q) \right) \wedge \left( \bigwedge_{r \in R} \lnot(0 = r) \right)
 \]
  be a ground formula over $\sigma_{or}(Y)$.
  Let $C$ be the least regular cone that contains $P$. If $F$ is
  satisfiable modulo $\ThQ$, then $C = \CnX{F}{Y}$.
  \label{lem:conseq-for-conjunction}
\end{lemma}
\begin{proof}
  Assume $F$ is satisfiable. We have $C \subseteq \CnX{F}{Y}$ since $\CnX{F}{Y}$
  is a regular cone that contains $P$. We thus only need to show
  that $\CnX{F}{Y} \subseteq C$.  Since $F$ is satisfiable, we have $\ConeModel{C} \models F$ by Theorem~\ref{thm:min-model-for-conjunction}.  For any $q \in \CnX{F}{Y}$, we have
  $F \models_{\ThQ} 0 \leq q$ and so $\ConeModel{C} \models 0 \leq q$ and
  therefore $q \in C$.
\end{proof}

This lemma justifies the assertion in Algorithm~\ref{alg:consequence} that we can implement a procedure $\texttt{get-model}(G)$ that returns a pair $\tuple{Z,P}$ where $\regcone(Z,P) = \CnX{G_i}{Y}$ for some cube
$G_i$ of the DNF of $G$.  The DPLL($\mathcal{T}$) algorithm samples some (propositionally) satisfiable cube $G_i$ of $G$ and decides if $G_i$ is $\ThQ$-satisfiable using
the procedure outlined in Section~\ref{sec:sat-mod-Q}; if not, then the cube $G_i$ is blocked, and we continue; if yes, then the model that is produced is the least cone containing all the non-negative atoms of $G_i$.  Since $G$ is $F$ conjoined with negated atoms,  this is also the least cone containing all of the non-negative atoms of some cube of $F$.

\subsubsection{Projection of algebraic cones} \label{sec:subsub-proj-alg-cones}

This section addresses the following problem: given finite sets $Z,P \subseteq \mathbb{Q}[Y]$ and a subset $X \subseteq Y$, compute $Z',P' \subseteq \mathbb{Q}[X]$ such that $\regcone(Z',P') = \regcone(Z,P) \cap \mathbb{Q}[X]$.  First, we review standard methods for solving this problem for ideals and finitely-generated cones (that is, the case when $P$ is empty, and the case where $Z$ is empty).  The algorithm for algebraic cones is a combination of the two.

Any monomial $m$ over variables $Y$ can be regarded as the product of two monomials $m_Xm_{\overline{X}}$ where $m_X$ is a monomial over $X$ and $m_{\overline{X}}$ is a monomial over $Y \setminus X$.
For any monomial ordering $\preceq$, we may define an \textit{elimination ordering} $\preceq_{X}$ where $m_Xm_{\overline{X}} \preceq_X n_Xn_{\overline{X}}$ iff $m_{\overline{X}} \prec n_{\overline{X}}$ or $m_{\overline{X}} = n_{\overline{X}}$ and $m_{X} \preceq n_{X}$.  The classical algorithm for ideal projection computes a basis for $\product{\mathbb{Q}[Y]}{Z} \cap \mathbb{Q}[X]$ by computing a Gr\"{o}bner basis $G$ for $Z$ with respect to the order $\preceq_X$, and then taking 
$G \cap \mathbb{Q}[X]$ \cite[Ch. 3]{Book:CLO2015}.
By the \textit{ordering} and \textit{membership} properties of Gr\"{o}bner bases, if $p \in \product{\mathbb{Q}[Y]}{Z} = \product{\mathbb{Q}[Y]}{G}$, then $p = q_1g_1 + \dots + q_ng_n$ for some $q_1,\dots,q_n \in \mathbb{Q}[Y]$ and $g_1,\dots,g_n \in G$ with $\leadmon(q_ig_i) \preceq_X \leadmon(p)$ for all $i$.  Supposing that $p$ is also in $\mathbb{Q}[X]$, each $q_i$ and $g_i$ must also be in $\mathbb{Q}[X]$.

We now turn to the case of finitely-generated cones.    Since $P$ is a finite collection of polynomials, there is a finite set of monomials that appear in any polynomial in $P$, which we call $M$.  Then we can see $\cone{P} \cap \mathbb{Q}[X] = \cone{P} \cap \product{\mathbb{Q}}{M \cap [X]}$; and so our problem is to compute the intersection of a cone and a linear space---this is the \textit{dual view} of the problem solved by polyhedral projection, for which there are several known algorithms.  For the sake of completeness, we will describe how to apply Fourier-Motzkin elimination, which is one such algorithm.  Suppose that we wish to compute $\cone{P} \cap \product{\mathbb{Q}}{N}$ for some set of monomials $N$.  Fourier-Motzkin elimination proceeds by eliminating one dimension, a monomial $m \in M \setminus N$, at a time.  First, we may normalize $P$ so that every
polynomial in $P$ takes the form $p + am$, where $m$ does not appear in $p$, and $a$ is either 0, 1, or -1 (by multiplying each polynomial with an appropriate non-negative scalar).  Then, take $Q = \{ p : p + 0m \in P \} \cup \{ p + q : p + m \in P, q - m \in P \}$. $\cone{Q}$ is precisely $\cone{P} \cap \product{\mathbb{Q}}{M \setminus m}$, since any non-negative combination of elements of $P$ that results in a coefficient of 0 for $m$ is also a non-negative combination of elements of $Q$.
Repeating this process for each monomial in $M \setminus N$, we get a finite set of polynomials that we denote $\polyhedralproject{P}{N}$, with $\cone{\polyhedralproject{P}{N}} = \cone{P} \cap \product{\mathbb{Q}}{N}$.

Finally, we put the two pieces together, by observing that we can project an algebraic cone $\regcone(Z,P)$ by separately projecting the ideal $\qideal{Z}$ and cone $\cone{P}$, \textit{provided that $\tuple{Z,P}$ is reduced with respect to the elimination ordering}.  That is, we define
\[ \project_X(Z,P) \defeq \tuple{ G \cap \mathbb{Q}[X], \polyhedralproject{\{ \red_G(p) : p \in P \}}{[X]} } \]
where $G$ is a Gr\"{o}bner basis for $\qideal{Z}$ w.r.t. the elimination order $\preceq_X$.

\begin{lemma}[Decomposition] \label{lem:decomp-to-smaller-monos}
  Let $Z, P \subseteq \mathbb{Q}[X]$ be finite sets of polynomials 
  such that $\tuple{Z,P}$ is
  reduced w.r.t. a monomial ordering
  $\preceq$. Suppose that
  $q \in \regcone(Z,P)$ is written as $q = z + p$ for some $z \in \qideal{Z}$
  and $p \in \cone{P}$. Then $\leadmon(z) \preceq \leadmon(q)$ and $\leadmon(p) \preceq \leadmon(q)$.
\end{lemma}
\begin{proof}
  Write $q \in \regcone(Z, P)$ as $z+p$ for $z \in \qideal{Z}$
  and $p \in \cone{P}$.

  For a contradiction, suppose that the lemma does not hold.
  Since $q = z + p$, the terms in $z$ and $p$ whose monomials are
  larger than $\leadmon(q)$ must cancel out.
  Thus we have $z = z' + t$ and $p = p' - t$ for some non-zero $t$ such that
  $\leadmon(z') \preceq \leadmon(q)$,
  $\leadmon(p') \preceq \leadmon(q)$,
  and $\leadmon(q) \prec \leadmon(t)$.
  Since
  $\leadmon(z') \preceq \leadmon(q) \prec \leadmon(t)$,
  we have
  $\leadmon(z) = \leadmon(z' + t) = \leadmon(t)$.
  Similarly, $\leadmon(p) = \leadmon(p' - t) = \leadmon(t)$.
  Thus $\leadmon(z) = \leadmon(t) = \leadmon(p)$.
  Since $z \in \qideal{Z}$, we must have $\red_Z(z) = 0$, and so
  $\leadmon(z)$ is divisible by $\leadmon(g)$ for some $g \in Z$.
  On the other hand, each $p_i \in P$ is reduced with respect to
  $Z$; since $p \in \cone{P}$, we have that
  $\leadmon(p) = \leadmon(z)$ is \textit{not} divisible by $\leadmon(g)$ for all $g \in Z$.
  This is a contradiction.
\end{proof}

\begin{theorem}[Projection] \label{thm:projection}
  Let $Z, P \subseteq \mathbb{Q}[Y]$ be finite sets of polynomials, and let $X \subseteq Y$.  Then
  \[ \algconeX{\project_X(Z,P)}{X} = \algconeX{Z,P}{Y} \cap \mathbb{Q}[X]\ . \]
\end{theorem}
\begin{proof}
  Let $G$ be a Gr\"{o}bner basis for $\product{\mathbb{Q}[Y]}{Z}$
  w.r.t. the elimination order $\preceq_X$, let
  $Z' = G \cap \mathbb{Q}[X]$,
  and let
  $P' = \polyhedralproject{\red_G(P)}{[X]}$,
  where $\red_G(P)$ denotes the set $\set{ \red_G(p) : p \in P }$.
  By Lemma~\ref{lem:ordering}, we have
  $\algconeX{Z, P}{Y} = \algconeX{G, \reduce_G(P)}{Y}$.
  We only need to show that
  $\algconeX{Z', P'}{X} = \algconeX{G, \reduce_G(P)}{Y} \cap \mathbb{Q}[X]$.
  We know that $\algconeX{Z', P'}{X} \subseteq \algconeX{G, \reduce_G(P)}{Y} \cap \mathbb{Q}[X] $
  since $G \cap \mathbb{Q}[X]$ is a basis for
  $\product{\mathbb{Q}[Y]}{G} \cap \mathbb{Q}[X]$ and
  $\cone{\polyhedralproject{\reduce_G(P)}{[X]}} = \cone{\reduce_G(P)} \cap \mathbb{Q}[X]$.
  Thus we only need to show the other direction.

  Consider any $y \in \algconeX{G, \reduce_G(P)}{Y} \cap \mathbb{Q}[X]$,
  i.e., $y = z + p$ where $z \in \product{\mathbb{Q}[Y]}{G}$ and $p \in \cone{\reduce_G(P)}$.
  By \autoref{lem:decomp-to-smaller-monos} we know that
  $\leadmon(z) \preceq_X \leadmon(y)$ and $\leadmon(p) \preceq_X \leadmon(y)$.
  Since $y \in \mathbb{Q}[X]$, we must have
  $z \in  \product{\mathbb{Q}[Y]}{G}  \cap \mathbb{Q}[X]$
  and $p \in \cone{ \reduce_G(P)} \cap \mathbb{Q}[X]$.
  Thus every $y \in \regcone(G,  \reduce_G(P))$ can be written as $z + p$ where
  $z \in \product{\mathbb{Q}[Y]}{G}  \cap \mathbb{Q}[X] = \product{\mathbb{Q}[X]}{Z'}$
  and $p \in \cone{ \reduce_G(P)} \cap \mathbb{Q}[X] = \cone{P'}$,
  hence $y = z + p \in \algconeX{Z', P'}{X}$.
\end{proof}



\subsubsection{Intersection of algebraic cones} \label{sec:intersection}

This section addresses the following problem: given finite sets $Z_1,P_1,Z_2,P_2 \subseteq \mathbb{Q}[X]$, compute $Z,P \subseteq \mathbb{Q}[X]$ such that $\regcone(Z,P) = \regcone(Z_1,P_1) \cap \regcone(Z_2,P_2)$.
The essential idea is to reduce the problem to a projection problem---essentially the same idea as the standard algorithm for ideal intersection \cite[Ch. 4 \S 3]{Book:CLO2015} and the constraint-based algorithm for polyhedral join \cite{TPLP:BKM2005}.\footnote{Recalling that cones are dual to polyhedra, cone intersection corresponds to polyhedral join, and cone sum to polyhedral meet.  The cone sum (like polyhedral meet) is easy: $\regcone(Z_1,P_1) + \regcone(Z_2,P_2) = \regcone(Z_1 \cup Z_2, P_1 \cup P_2)$.}

The essential idea is to introduce a parameter $t$ that does not belong to $X$, and to ``tag'' each element of a cone $C_1 = \regcone(Z_1,P_1)$ by multiplying by $t$, and to tag elements of $C_2 = \regcone(Z_2,P_2)$ by multiplying by $1-t$.
If $p$ is a polynomial that belongs to $C_1$ and $C_2$, then $tp + (1-t)p = p$ belongs to their ``tagged sum.''
This yields the following definition:
\[ \textit{intersect}(Z_1,P_1,Z_2,P_2) \defeq \project_X(tZ_1 \cup (1-t)Z_2, tP_1 \cup (1-t)P_2) \]
where the notation $pQ \defeq \set{ pq : q \in Q}$ (for a polynomial $p$ and set of polynomials $Q$) denotes the ``tagging'' operation.

\begin{mexample}
  Consider the regular cones
  \[
  \begin{array}{rcccl}
  C_1 &=& \Cn{x = 1 \land y \leq 1} &=& \regcone(\set{x-1},\set{1, 1-y})\\
  C_2 &=& \Cn{y = 2 \land 2 \leq x^2 } &=& \regcone(\set{y-2},\set{1,x^2-2})
  \end{array}
  \]
  To intersect $C_1$ and $C_2$, we form
  $Z = \set{t (x-1), (1-t)(y-2)}$
  and $P = \set{t, t(1-y), (1-t), (1-t)(x^2-2)}$.
  To compute a Gr\"{o}bner basis $G$
  (w.r.t $\preceq_{\set{x,y}}$) for $\product{\mathbb{Q}[x,y,t]}{Z}$,
  we start with $Z = \set{tx-t, -ty+2t+y-2}$ and
  apply Buchberger's algorithm to complete it to
  $G = \set{tx-t, -ty+2t+y-2, xy-2x-y+2}$.
  Then reduce $P$ to get $\reduce_G(P) = \set{t, -t-y+2, -t + 1, t + x^2 - 2 }$.
  Then intersecting $G$ with $\mathbb{Q}[x,y]$ and projecting the monomial $t$ out of
  $\reduce_G(P)$ using Fourier-Motzkin elimination, we get
   \[
   C_1 \cap C_2 = \regcone(\set{xy-2x-y+2},\set{1, -y+2, x^2-y, x^2 - 1})\ ,
   \]
   which is the non-negative cone of the formula
   $(x-1)\cdot (y-2) = 0 \land 2 \geq y \land x^2 \geq y \land x^2 \geq 1$.
\end{mexample}

To prove correctness of this construction, we need the following technical lemma relating cones to their ``tagged'' counterparts:

\begin{lemma}
\label{lem:tag}
  Let $X$ be a set of variables and $t \notin X$. Let $Z, P \subseteq \mathbb{Q}[X]$ be
  finite sets of polynomials, and let $f \in \mathbb{Q}[t]$ be a polynomial in $t$.
  Then for all $a \in \mathbb{Q}$ such that $f(a) \geq 0$, and for all
  $q \in \algconeX{fZ,fP}{X,t}$, we have $q[t \mapsto a] \in \algconeX{Z,P}{X}$
  (where $q[t \mapsto a]$ denotes substitution of occurrence of $t$ in $q$ with $a$).
\end{lemma}
\begin{proof}
  Let $Z = \set{z_1, \dots, z_m}$ and $P = \set{p_1, \dots, p_n}$.
  For any $q \in \algconeX{fZ,fP}{X,t}$ we can write
  \begin{align*}
      q &= \sum_{i=1}^m g_ifz_i + \sum_{j=1}^n \lambda_jfp \quad (\forall i.\,q_i \in \mathbb{Q}[X,t], \lambda_i \in \mathbb{Q}^{\geq 0}) \\
      q[t \mapsto a] &= \sum_{i \in 1}^m \left(g_if(a)\right) z + \sum_{j=1}^n \left(\lambda_{j} f(a)\right) p \in \regcone(Z,P)
  \end{align*}
  since each $g_if(a)$ is a polynomial in $X$ and each
  $\lambda_jf(a)$ is a non-negative rational.
\end{proof}

\begin{theorem}[Intersection] \label{thm:intersection}
  Let $Z_1, P_1, Z_2, P_2 \subseteq \mathbb{Q}[X]$ be finite sets of polynomials over some set of variables $X$.
  Then 
  \[ 
  \regcone(\textit{intersect}(Z_{1}, P_{1}, Z_{2}, P_{2})) = \regcone(Z_1, P_1) \cap \regcone(Z_2, P_2)\ .
  \]
\end{theorem}
\begin{proof}
We prove each side of the equation is included in the other:
\begin{itemize}
    \item[$\subseteq$]:
  Let $q \in \regcone(\textit{intersect}(Z_1,P_1,Z_2,P_2))$.  
  Since $\textit{intersect}(Z_1,P_1,Z_2,P_2) = \textit{project}_X(tZ_1 \cup (1-t)Z_2, tP_1 \cup (1-t)P_2)$, we have $q \in \regcone(tZ_1 \cup (1-t)Z_2, tP_1 \cup (1-t)P_2) \cap \mathbb{Q}[X]$
  by Theorem~\ref{thm:projection}.
  Then
  $q$ can be written as $q_1 + q_2$ for some $q_1 \in \regcone(tZ_1,tP_1)$ and $q_2 \in \regcone((1-t)Z_2,(1-t)P_2)$.
  Then we have
  \begin{align*}
      q &= q[t \mapsto 0] & q \in \mathbb{Q}[X]\\
      &= q_1[t \mapsto 0] + q_2[t \mapsto 0] & q = q_1+q_2, \text{linearity of substitution}\\
      &=q_2[t \mapsto 0] & t \text{ divides } q_1\\
      &\in \regcone(Z_{2}, P_{2}) & Lemma~\ref{lem:tag}
  \end{align*}
  Symmetrically, we have $q = q[t \mapsto 1] = q_1[t \mapsto 1] \in \regcone(Z_1,P_1)$, and so $q$ belongs to the intersection
  $\regcone(Z_{1}, P_{1}) \cap \regcone(Z_{2}, P_{2})$.

    \item[$\supseteq$]:
  Let $q \in \regcone(Z_{1}, P_{1}) \cap \regcone(Z_{2}, P_{2})$.
  We have $tq \in \regcone(tZ_{1}, tP_{1})$ and $(1-t)q \in \regcone((1-t)Z_{2}, (1-t)P_{2})$,
  and therefore
  \begin{align*}
  q = tq + (1-t)q &\in
  \algconeX{tZ_1,tP_1}{X,t}
  + \algconeX{(1-t)Z_2,(1-t)P_2}{X,t}\\
  &= 
  \algconeX{tZ_1\cup (1-t)Z_2, tP_1 \cup (1-t)P_2}{X,t}\ .
  \end{align*}
  Since $q$ also belongs to $\mathbb{Q}[X]$, we have
 $q \in \algconeX{tZ_1\cup (1-t)Z_2, tP_1 \cup (1-t)P_2}{X,t} \cap \mathbb{Q}[X] = \regcone(\textit{intersect}(Z_1,P_1,Z_2,P_2))$. \qedhere
\end{itemize}

\end{proof}

\subsubsection{Correctness of consequence finding}
\label{sec:correctness-conseq-finding}
We are now ready to prove the correctness of consequence finding.

\begin{theorem}
  Given a $\sigma_{or}(Y)$ formula $F$ and $X \subseteq Y$, 
  $\texttt{consequence}(F, X)$ 
  terminates and returns a reduced pair $\tuple{Z, P}$ such that
  $\algconeX{Z, P}{X} = \CnX{F}{X}$.
\end{theorem}
\begin{proof}

  Say that a regular cone $C$ is a
  \textit{cube cone} of formula $H$ if $C = \CnX{H}{X}$ for some cube $D$ in
  the DNF of $H$, and $C$ is consistent. Let $\mathit{cube.cone}(H)$ denote the (finite) set of cube cones of $H$.
  
  We first prove termination of the algorithm. 
  Let $G_i$ denote the formula $G$ on the $i$th iteration of the loop,
  and let $B_i$ denote the $i$th blocking clause, so for each $i$ we have $G_{i+1} = G_i \land \lnot B_i$.
  By Theorem~\ref{thm:min-model-for-conjunction} and Lemma~\ref{lem:consequence-cone}, we have that
  \[\mathit{cube.cones}(G_{i+1}) = \mathit{cube.cones}(G_i \land \lnot B_i) = \set{Q \in \mathit{cube.cones}(G_i) : \ConeModel{Q} \models \lnot B_i}\ .\]
  Since the model returned by
  \texttt{get-model} is always in $\mathit{cube.cone}(G_i)$ (Theorem~\ref{thm:min-model-for-conjunction} and Lemma~\ref{lem:consequence-cone}) and that model always satisfies $B_i$, we have a strictly descending sequence
  $\mathit{cube.cones}(G_{0}) \supsetneqq \mathit{cube.cones}(G_{1}) \supsetneqq \mathit{cube.cones}(G_{2}) \supsetneqq \dots$.
  Since $\mathit{cube.cones}(G_{0})$ is a finite set, this sequence must have finite length and therefore the algorithm terminates.
  
  We then show $\algconeX{Z, P}{X} = \CnX{F}{X}$ when the while loop exits.
  First we prove one direction that $\algconeX{Z, P}{X} \supseteq \CnX{F}{X}$.
  Suppose the while loop runs $N$ times.
  By Theorems~\ref{thm:projection} and~\ref{thm:intersection}, 
  $\algconeX{Z,P}{X} = \mathbb{Q}[X] \cap \bigcap_{i = 1}^N C_i$
  where each $C_i \in \mathit{cube.cone}(F)$, thus
  $\algconeX{Z, P}{X} \supseteq \CnX{F}{X}$.
  We then show that $\algconeX{Z, P}{X} \subseteq \CnX{F}{X}$.
  When the while loop exits, we have that
  $G_N = F \land \neg B_{1} \land \dots \land \neg B_{N} $ is unsatisfiable modulo $\ThQ$. 
  Since $B_1 \models_{\ThQ} B_2 \models_{\ThQ} \dots \models_{\ThQ} B_N$, we have $F \models_{\ThQ} B_{N}$, and
  therefore $\CnX{F}{X} \supseteq \CnX{B_{N}}{X} = \algconeX{Z, P}{X}$.

  

  
\end{proof}

%% file: integer.tex
\section{Integer Arithmetic}
\label{sec:integer}
\newcommand{\Int}{\textit{Int}}

Let $\sigma_{or}^Z$ be the signature of ordered rings extended with an
additional unary relation symbol $\Int$.
Define the theory of \textit{linear integer real rings}
$\ThZ$ to be the $\sigma_{or}^Z$-theory axiomatized by the axioms
of $\ThQ$ along with the following:
\[\begin{array}{cr}
    \Int(1)\\
    \forall x,y\ldotp \Int(x) \land \Int(y) \Rightarrow \Int(x+y) & \text{(Int closure +)}\\
    \forall x,y\ldotp \Int(x) \land x + y = 0 \Rightarrow \Int(y) & \text{(Int closure -)}\\
    \text{for all } n \in \mathbb{Z}^{>0} \text{ and }  m \in \mathbb{Z},
    \forall x\ldotp \Int(x) \land 0 \leq nx + m
    \Rightarrow 0 \leq x + \floor{\frac{m}{n}} & \text{(Cutting plane)}
\end{array}\]

We regard $\mathbb{R}$ as the ``standard'' model of $\ThZ$,
with $\Int$ identifying the subset of integers; $\Th^Z(\mathbb{R})$ refers
to the theory consisting of all $\sigma_{\textit{or}}^Z$-sentences satisfied by $\mathbb{R}$.
As before, the theory admits a class of non-standard models that captures the
characteristics of all models of $\ThZ$.

\zak{Can be more specific: To extend our construction of $\ThQ$-models from
  regular cones to the theory $\ThZ$, we must develop (1) a notion of a
  suitable interpretation of the $\Int$ predicate, and (2) a criterion under
  which the cutting plane axiom is satisfied.  Then we can get into the
  next paragraph, where the reader knows what to expect.  Curiously, however,
  it polynomial lattices don't seem like \textit{the} solution -- the Int
  predicate should just be a an additive subgroup.  Polynomial lattices are an
  analogue of algebraic cones (lattices that can be effectively manipulated).
  If we have frame things in terms of additive subgroups, then we have an
  direct analog of Lemma~\ref{lem:model-construction} (without needing ``least'' models)}
  
Let $X$ be a set of variables.
For sets $C, L \subseteq \mathbb{Q}[X]$, say that $C$ is
\textit{closed under cutting planes with respect to $L$} if
if for any $a, b \in \mathbb{Z}$, $a > 0$, $p \in L$,
if $ap + b \in C$, then $p + \floor{\frac{b}{a}} \in C$.
Define the \textit{cutting plane closure of $C$ with respect to $L$}, denoted $\cp_L(C)$, to be the least cone that
contains $C$ and is closed under cutting planes with respect to $L$.
A set $L \subseteq \mathbb{Q}[X]$ is a \textbf{polynomial lattice}
if $L = I + L_0$ for some ideal $I$ and a point lattice $L_0$.
If $C$ is a regular cone and $L$ is a polynomial lattice,
the pair $\tuple{C, L}$ is \textbf{coherent} if $L = \units{C} + L_0$
for some point lattice $L_0$.

We associate a coherent
$\tuple{C, L}$ with a $\sigma_{or}^Z(X)$-structure
$\ConeModel{C, L}$, where the universe and interpretations of
all symbols other than $\Int$ follow that of $\ConeModel{C}$,
and $\Int$ is interpreted as the relation $\set{p + \units{C} : p \in L}$.
(This interpretation is well-defined because $\tuple{C, L}$ is coherent.)
That is, for all $q \in \mathbb{Z}[X]$,
$\ConeModel{C, L} \models \Int(q)$ iff $q \in L$.

\begin{theoremEnd}[restate]{lemma}
\label{lem:thz-model-construction}
  Let $X$ be a set of variables.
  Let $\tuple{C, L}$ be coherent with $1 \in L$ and
  $C$ consistent and closed under cutting planes with respect to $L$.
  Then $\ConeModel{C, L}$ is a model of $\ThZ$.
\end{theoremEnd}
\begin{proofEnd}
  Since $1 \in L$, $1 + I \in \Int^{\ConeModel{C, L}}$ and
  $\ConeModel{C, L} \models \Int(1)$.
  The closure axioms hold because $L$ is closed under addition and taking
  additive inverses.

  For the cutting plane axiom,
  let $n, m \in \mathbb{Z}$, $n > 0$,
  and $p \in \mathbb{Z}[X]$ be such that $np + m \in C$ and $p \in L$.
  Since $C$ is closed under cutting planes with respect to $L$,
  $p + \floor{\frac{m}{n}} \in C$.
  So $\ConeModel{C, L} \models 0 \leq p + \floor{\frac{m}{n}}$.
\end{proofEnd}

\begin{theoremEnd}{theorem}
  \label{thm:ThZ-min-model-for-conjunction}
  Given a conjunctive ground formula
  \[
  F = \left(\bigwedge_{p \in P} 0 \leq p \right) \land
  \left( \bigwedge_{q \in Q} \lnot (0 \leq q) \right) \land
  \left( \bigwedge_{r \in R} \lnot(0 = r) \right) \land
  \left( \bigwedge_{s \in S} \Int(s) \right) \land
  \left( \bigwedge_{t \in T} \lnot \Int(t) \right)
  \]
  Let $B = S \cup \set{1}$.
  Let $C$ be the least regular cone that contains $P$
  and is closed under cutting planes with respect to $\units{C} + \product{\mathbb{Z}}{B}$.
  Then $F$ is satisfiable if and only if $C$ is consistent and
  $\ConeModel{C, \units{C} + \product{\mathbb{Z}}{B}} \models F$.
\end{theoremEnd}
\begin{proofEnd}
  Sufficiency follows from Lemma~\ref{lem:thz-model-construction}.
  For necessity, assume $F$ to be satisfiable, and consider some $\ThZ$-model
  $M$ that satisfies $F$.
  Following the proof of Theorem~\ref{thm:min-model-for-conjunction},
  $\CnM{M} = \set{p \in \mathbb{Q}[X]: M \models 0 \leq p}$ is a 
  regular cone that contains $P$.

  Consider $np + m \in \CnM{M}$, where $n, m \in \mathbb{Z}$, $n > 0$,
  and $p \in \product{\mathbb{Z}}{B}$.
  Since $M \models \Int(s)$ for all $s \in S$ and $M \models \Int(1)$,
  by the $\Int$ closure axioms,
  $M \models \Int(p)$.
  Since $M$ satisfies the cutting plane axiom,
  $M \models p + \floor{\frac{m}{n}} \geq 0$.
  Hence, $p + \floor{\frac{m}{n}} \in \CnM{M}$.

  Thus, $C \subseteq \CnM{M}$.
  As per the proof of Theorem~\ref{thm:min-model-for-conjunction},
  $C$ is consistent, and
  $\ConeModel{C, \units{C} + \product{\mathbb{Z}}{B}}$
  satisfies $0 \leq p$, $\lnot (0 \leq q)$, and $\lnot (0 = r)$
  for all $p \in P$, $q \in Q$ and $r \in R$.
  It satisfies all $\Int(s)$ since each $s \in B$.
  If $\ConeModel{C, \units{C} + \product{\mathbb{Z}}{B}} \models \Int(t)$
  for some $t \in T$,
  then $t \in \units{C} + \product{\mathbb{Z}}{B}$,
  i.e., $t = z + b$ for some $z \in \units{C}$ and
  $b \in \product{\mathbb{Z}}{B}$.
  Since $C \subseteq \CnM{M}$, $M \models z = 0$.
  Since $M \models \Int(0)$ and $M \models \Int(b)$,
  $M \models \Int(t)$, contradicting $M \models F$.
  So $\ConeModel{C, \units{C} + \product{\mathbb{Z}}{B}} \models \lnot \Int(t_k)$.
  Thus, it satisfies $F$.
\end{proofEnd}

An induction axiom is conspicuously absent from $\ThZ$.  Nevertheless, the axiomatization is sufficient for positive linear formulas.

\begin{theorem} \label{thm:lirr-lira-complete}
  Let $F$ be a ground formula in $\sigma_{or}^Z(X)$ that is free of negation and
  multiplication.  Then $F$ is satisfiable modulo $\ThZ$ iff $F$ is
  satisfiable modulo $\Th^Z(\mathbb{R})$.
\end{theorem}
\begin{proof}
  The $\Leftarrow$ direction  is trivial, since any model of $\Th^Z(\mathbb{R})$ is a model of  $\ThZ$.

  For the $\Rightarrow$ direction, we prove that if $F$ is unsatisfiable modulo $\Th^Z(\mathbb{R})$, then it is unsatisfiable modulo $\ThZ$.
  We may assume without loss of generality that $F$ takes the form
  $\bigwedge_{p \in P} p \geq 0 \land \bigwedge_{s \in S} \Int(s)$.
  We may also assume that $S \subseteq X$:
  if $F = G \land \Int(t)$ for some non-constant term $t$,
  then $F$ is equisatisfiable with $G \land \Int(y) \land y \leq t \land t \leq y$ for a fresh constant $y$
  (modulo $\Th^Z(\mathbb{R})$ and also modulo $\ThZ$).
  Furthermore, we may assume that $S = X$ (i.e., all symbols are integral),
  since if some symbol is not constrained to be integral, it can be projected
  by Fourier-Motzkin elimination, resulting in an equisatisfiable formula
  (again modulo both theories).



  Suppose that $F$ is unsatisfiable modulo $\Th^Z(\mathbb{R})$---i.e., the polyhedron defined by $\bigwedge_{p \in P} p \geq 0$ has no integer points.
  Then there is a cutting-plane proof of
  $0 \leq -1$ from $\bigwedge_{p \in P} p \geq 0$~\cite{Chvatal1973,Schrijver1980}.
  Since each inference step of a cutting-plane proof is valid modulo $\ThZ$,
  $F$ is unsatisfiable modulo $\ThZ$.
\end{proof}

If $\tuple{C, L}$ is coherent, we may represent it using
a triple $\tuple{Z, P, B}$, where
$Z, P, B \subseteq \mathbb{Q}[X]$ are finite sets,
$C = \regcone(Z,P)$, and
$L = \qideal{Z} + \product{\mathbb{Z}}{B}$.
Say that $\tuple{Z, P, B}$ is \textbf{reduced} (with respect to a monomial ordering
$\preceq$ if
$\tuple{Z, P}$ is reduced (with respect to $\preceq$)), each
$b \in B$ is reduced with respect to $Z$,
and $B$ is linearly independent.

\zak{Needed some context -- why do we care about membership?}
Extending our earlier results for $\ThQ$, we have that the problem of checking
whether a model $\ConeModel{\regcone(Z,P), \qideal{Z}+B}$ satisfies a ground
formula $F$ is decidable.  This follows from Lemma~\ref{lem:membership} and the following.

\zak{Analogue of the ordering lemma (i.e., we can always compute a reduced representation)}

\begin{lemma}[Membership]
  \label{lem:lattice-membership}
  If $\tuple{Z, P, B}$ is reduced, for any polynomial $p \in \mathbb{Q}[X]$,
  we have $p \in \qideal{Z} + \product{\mathbb{Z}}{B}$
  iff $\reduce_Z(p) \in \product{\mathbb{Z}}{B}$.
  The latter can be checked in polytime by solving the linear system for its
  unique solution and checking if all coordinates are integers.
\end{lemma}

\subsection{Satisfiability modulo $\ThZ$} \label{sec:sat-mod-Z}

\zak{Inform the reader of the high-level structure of the plan.}

This section presents a theory solver for testing satisfiability of ground
conjunctive $\sigma_{or}^Z$ formulas modulo the theory $\ThZ$. Following
  Theorem~\ref{thm:ThZ-min-model-for-conjunction}, \zak{or just `` the strategy for $\ThQ$ '' if we get rid of minimal models in that theorem}
  we show that it is possible to construct
  \textit{minimal} $\ThZ$-models of ground conjunctive $\sigma_{or}^Z$
  formulas.
  We first show how to compute cutting plane closure (Lemma~\ref{lem:cut-correct}).
  Then we show that minimal models can be constructed by
  iterating the regular
  closure (that is, Algorithm~\ref{alg:saturation}) and
  cutting plane closure separately until fixed point (Lemma~\ref{lem:saturate-cp-closure}).

\zak{Consider re-ordering: first cutting plane closure, then tying the two
  together.  It's intuitive that the two can be iterated to a simultaneous
  fixpoint to get closure under both (Knaster-Tarski).  The thing that isn't
  obvious is termination.  Lemma~\ref{lem:cp-ineq-invariance} can get rolled into
  Theorem~\ref{thm:rcp}}
\nic{Yep, thought cutting plane closure should come first too; rearranged.
  I'm leaving Lemma~\ref{lem:cp-ineq-invariance} in for now because it explains
  the shift from polynomial lattices to point lattices.
}

\zak{Give the reader a roadmap for how things fit together}
To compute the cutting plane closure of an algebraic cone with respect
to a (coherent) polynomial lattice, we proceed in three steps.  First,
Lemma~\ref{lem:cp-ineq-invariance} shows that it is sufficient to consider
point (rather than polynomial) lattices.  Second,
Lemma~\ref{lem:cp-ineq-closure} shows that cutting plane closure of an
algebraic cone with respect to a point lattice can be reformulated in
terms of the cutting plane closure of a polyhedral cone with respect to a
point lattice. Third, we show that these lemmas constitute an effective
reduction from the problem of computing the cutting plane closure of an
algebraic cone with respect to a polynomial lattice to that of computing the
integer hull of a polyhedron.

\begin{lemma}
  \label{lem:cp-ineq-invariance}
  Let $C, L \subseteq \mathbb{Q}[X]$ with $C$ a cone.
  Let $L' \subseteq \units{C} + L$.
  Then $\cp_{L'}(C) \subseteq \cp_{L}(C)$.
  Thus, $\cp_{\units{C} + L}(C) = \cp_{L}(C)$.
\end{lemma}
\begin{proof}
  First observe that $\cp_{L}(C)$
  contains $C$, so it suffices to show that it is closed under cutting
  planes with respect to $L'$.

  Let $np + m \in \cp_{L}(C)$,
  where $n, m \in \mathbb{Z}$, $n > 0$, and $p \in L'$.
  By definition, $p = z + v$ for some $z \in \units{C}$ and $v \in L$.
  Since $nz \in \units{C} \subseteq \units{\cp_{L}(C)}$,
  $nv + m = n(z + v) + m - nz = (np + m) - nz \in \cp_{L}(C)$.
  Since $\cp_{L}(C)$ is closed under cutting planes with respect to $L$
  and $v \in L$, $v + \floor{\frac{m}{n}} \in \cp_{L}(C)$.
  Then $p + \floor{\frac{m}{n}} = z + v + \floor{\frac{m}{n}} \in \cp_{L}(C)$.
  Thus, $\cp_{L'}(C) \subseteq \cp_{L}(C)$.

  Setting $L' = \units{C} + L$, $\cp_{\units{C} + L}(C) \subseteq \cp_{L}(C)$.
  Since $L \subseteq L'$, $\cp_{L}(C) \subseteq \cp_{\units{C} + L}(C)$,
  so we have equality.
\end{proof}

The following is the key result of this section.  It shows that we may
reduce the problem of computing the cutting plane closure of an algebraic cone
with respect to a point lattice to the problem of computing the integer
hull of a (finite-dimensional) polyhedron.

\begin{lemma}
  \label{lem:cp-ineq-closure}
  Let $B = \set{b_1, \ldots, b_k} \subseteq \mathbb{Q}[X]$,
  and let $C \subseteq \mathbb{Q}[X]$ be a cone.
  Let $Y = \set{y_1, \ldots, y_k}$ be a set of variables
  disjoint from $X$ and define a linear map
  $f: \linpoly{Y} \to \mathbb{Q}[X]$ by
  \[
  f(a_0 + a_1y_1 + \dots + a_ky_k) = a_0 + a_1b_1 + \dotsi +a_kb_k.
  \]
  Then $\cp_{\, \product{\mathbb{Z}}{B}}(C) = C + f(\cp_{\, \product{\mathbb{Z}}{Y}}(f^{-1}(C)))$.
\end{lemma}
\begin{proof}
  ($\subseteq$)
  Since $\cp_{\, \product{\mathbb{Z}}{B}}(C)$ is the \textit{least} cone that contains $C$
  and is closed under cutting planes with respect to $\product{\mathbb{Z}}{B}$,
  it is sufficient to prove the following:
  \begin{enumerate}
  \item $C + f(\cp_{\, \product{\mathbb{Z}}{Y}}(f^{-1}(C)))$ is a cone: it is the sum of two cones and is thus a cone.
  \item $C + f(\cp_{\, \product{\mathbb{Z}}{Y}}(f^{-1}(C)))$ contains $C$: since $C$ is non-empty,
    $0 \in f(\cp_{\,\product{\mathbb{Z}}{Y}}(f^{-1}(C)))$,
    so $C \subseteq C + f(\cp_{\,\product{\mathbb{Z}}{Y}}(f^{-1}(C)))$.
  \item $C + f(\cp_{\, \product{\mathbb{Z}}{Y}}(f^{-1}(C)))$ is closed under cutting planes w.r.t. $\product{\mathbb{Z}}{B}$.
  Suppose $np + m \in C + f(\cp_{\product{\mathbb{Z}}{Y}}(f^{-1}(C)))$,
  with $n, m \in \mathbb{Z}$, $n > 0$, and $p \in \product{\mathbb{Z}}{B}$; we must show that
  $p + \floor{\frac{m}{n}} \in C + f(\cp_{\, \product{\mathbb{Z}}{Y}}(f^{-1}(C)))$.

  Without loss of generality, we may suppose $np + m \in f(\cp_{\product{\mathbb{Z}}{Y}}(f^{-1}(C)))$---the argument is as follows.
  Since $np+m \in C + f(\cp_{\product{\mathbb{Z}}{Y}}(f^{-1}(C)))$, we have
   $np + m = g  + h$ for some $g \in C$ and
  $h \in f(\cp_{\, \product{\mathbb{Z}}{Y}}(f^{-1}(C)))$.
  Since $np + m \in \product{\mathbb{Z}}{B} + \mathbb{Z} = f(\linpoly{Y})$, we have
  $g = (np + m) - h \in f(\linpoly{Y})$ since the image of $f$ is a subspace containing both $(np+m)$ and $h$.
  So $g \in C \cap f(\linpoly{Y})$, and
  thus $g \in f(f^{-1}(C))$.
  Since $\cp_{\,\product{\mathbb{Z}}{Y}}$ is extensive
  (i.e., $S \subseteq \cp_{\,\product{\mathbb{Z}}{Y}}(S)$ for all $S$)
  and $f$ is linear,
  $np + m = g + h \in f(\cp_{\, \product{\mathbb{Z}}{Y}}(f^{-1}(C)))$.

  We now prove that $p + \floor{\frac{m}{n}} \in C + f(\cp_{\, \product{\mathbb{Z}}{Y}}(f^{-1}(C)))$.    Since $p \in \product{\mathbb{Z}}{B}$, there exists $p' \in \product{\mathbb{Z}}{Y}$
  such that $f(p') = p$.
  Then $f(np' + m) = np + m \in f(\cp_{\, \product{\mathbb{Z}}{Y}}(f^{-1}(C)))$.
  Since $f$ is a linear map,
  $np' + m + g \in \cp_{\, \product{\mathbb{Z}}{Y}}(f^{-1}(C))$
  for some $g \in \ker(f)$.
  Since $0 \in C$, $\ker(f) = f^{-1}(0) \subseteq f^{-1}(C) \subseteq \cp_{\, \product{\mathbb{Z}}{Y}}(f^{-1}(C))$,
  so $-g \in \cp_{\, \product{\mathbb{Z}}{Y}}(f^{-1}(C))$.
  Hence, $np' + m \in \cp_{\, \product{\mathbb{Z}}{Y}}(f^{-1}(C))$.
  Since $p' \in \product{\mathbb{Z}}{Y}$,
  $p' + \floor{\frac{m}{n}} \in \cp_{\, \product{\mathbb{Z}}{Y}}(f^{-1}(C))$.
  So $p + \floor{\frac{m}{n}} = f(p' + \lfloor \frac{m}{n} \rfloor) \in C + f(\cp_{\, \product{\mathbb{Z}}{Y}}(f^{-1}(C)))$.
  \end{enumerate}

  ($\supseteq$) First note that $C \subseteq \cp_{\, \product{\mathbb{Z}}{B}}(C)$,
  and since $\cp_{\, \product{\mathbb{Z}}{B}}(C)$ is closed under addition, it suffices to show that
  $f(\cp_{\, \product{\mathbb{Z}}{Y}}(f^{-1}(C))) \subseteq \cp_{\, \product{\mathbb{Z}}{B}}(C)$.
  In turn, it suffices to show
  $\cp_{\, \product{\mathbb{Z}}{Y}}(f^{-1}(C)) \subseteq f^{-1}(\cp_{\, \product{\mathbb{Z}}{B}}(C))$.
  As before, it suffices to show that
  $f^{-1}(\cp_{\product{\, \mathbb{Z}}{B}}(C))$ is a cone,
  $f^{-1}(C) \subseteq f^{-1}(\cp_{\product{\, \mathbb{Z}}{B}}(C))$
  and $f^{-1}(\cp_{\product{\, \mathbb{Z}}{B}}(C))$ is also closed.
  It is easy to verify that $f^{-1}(\cp_{\product{\, \mathbb{Z}}{B}}(C))$ is a cone,
  and since $\cp_{\, \product{\mathbb{Z}}{B}}$ is extensive,
  $f^{-1}(C) \subseteq f^{-1}(\cp_{\product{\, \mathbb{Z}}{B}}(C))$.
  Suppose that $np + m \in f^{-1}(\cp_{\, \product{\mathbb{Z}}{B}}(C))$,
  with $p \in \product{\mathbb{Z}}{Y}$,
  $n, m \in \mathbb{Z}$, and $n > 0$; we must show that
  $p + \floor{\frac{m}{n}} \in f^{-1}(\cp_{\,\product{\mathbb{Z}}{B}}(C))$.
  Since $f$ preserves $1$, we have
  $n f(p) + m = f(np + m) \in \cp_{\, \product{\mathbb{Z}}{B}}(C)$.
  Since $f(p) \in \product{\mathbb{Z}}{B}$,
  and $\cp_{\, \product{\mathbb{Z}}{B}}(C)$
  is closed under cutting planes with respect to $\product{\mathbb{Z}}{B}$,
  $f(p) + \lfloor \frac{m}{n} \rfloor \in \cp_{\, \product{\mathbb{Z}}{B}}(C)$.
  Thus, $f(p + \floor{\frac{m}{n}}) \in \cp_{\, \product{\mathbb{Z}}{B}}(C)$,
  and $p + \floor{\frac{m}{n}} \in f^{-1}(\cp_{\,\product{\mathbb{Z}}{B}}(C))$.
\end{proof}

It remains to show that the reduction in Lemma~\ref{lem:cp-ineq-closure} is
effective.  The only non-trivial operation involved in the reduction that we
have not already seen is computing the inverse image of an algebraic cone
$C \subseteq \mathbb{Q}[X]$ under a linear map with finite-dimensional domain
$f: \linpoly{Y} \rightarrow \mathbb{Q}[X]$.
This can be accomplished by
extending the linear map to a ring homomorphism $\hat{f}: \mathbb{Q}[Y]
\rightarrow \mathbb{Q}[X]$, taking its inverse image, and intersecting it with
$\linpoly{Y}$.

For disjoint sets of variables $X$ and $Y$,
$f: \mathbb{Q}[Y] \to \mathbb{Q}[X]$ a ring homomorphism,
and $Z, P \subseteq \mathbb{Q}[X]$ with $Z$ a Gr\"{o}bner basis,
define
\[
\textit{inverse-hom}(Z,P,f,Y) \defeq \textit{project}_Y(\set{ y-f(y) : y \in Y } \cup Z, P).
\]

\begin{theoremEnd}[restate]{theorem}[Inverse image]
  \label{thm:inverse-ring-groebner}
  Let $Z, P \subseteq \mathbb{Q}[X]$ be finite with $Z$ a Gr\"{o}bner basis.
  Let $Y$ be a finite set of variables distinct from $X$.
  Let $f : \mathbb{Q}[Y] \to \mathbb{Q}[X]$ be a ring homomorphism.
  Then
  \[
  \algconeX{\textit{inverse-hom}(Z,P,f,Y)}{Y} = f^{-1}(\algconeX{Z,P}{X}).
  \]
\end{theoremEnd}
\begin{proofEnd}
  Let $T = \set{y - f(y) : y \in Y}$.
  Expanding the definition of \textit{inverse-hom} and using Theorem~\ref{thm:projection},
  $\algconeX{\textit{inverse-hom}(Z,P,f,Y)}{Y} = \algconeX{T \cup Z, P}{X, Y} \cap \mathbb{Q}[Y]$.
  
  ($\subseteq$)
  If $s \in \algconeX{T \cup Z, P}{X, Y} \cap \mathbb{Q}[Y]$, then
  $s - p \in \product{\mathbb{Q}[X, Y]}{T \cup Z}$ for some
  $p \in \product{\Qplus}{P}$.

  Extend the ring homomorphism $f: \mathbb{Q}[Y] \to \mathbb{Q}[X]$
  to a homomorphism $\hat{f}: \mathbb{Q}[X, Y] \to \mathbb{Q}[X]$
  such that $\hat{f}(x) = x$ for each variable $x \in X$
  and $\hat{f}(y) = f(y)$ for each variable $y \in Y$,
  so that $\hat{f} |_{\, \mathbb{Q}[Y]} = f$.
  By Lemma~\ref{lem:eval-ring-hom},
  and noting that $T = \set{y - f(y) : y \in Y} = \set{y - \hat{f}(y) : y \in Y}$,
  $\hat{f}(s-p) - (s - p) \in \product{\mathbb{Q}[X]}{T \cup Z}$.

  Then $\hat{f}(s - p) = (\hat{f}(s-p) - (s - p)) + (s - p) \in \product{\mathbb{Q}[X]}{T \cup Z}$.
  Since $Z \subseteq \mathbb{Q}[X]$ is a Gr\"{o}bner basis and $X \cap Y = \emptyset$,
  $T \cup Z$ is a Gr\"{o}bner basis for $\product{\mathbb{Q}[X]}{T \cup Z}$ 
  with respect to an elimination order $\preceq_X$.
  Since $\hat{f}(s - p) \in \mathbb{Q}[X]$,
  by the ordering property of Gr\"{o}bner bases, 
  $\hat{f}(s - p) \in \product{\mathbb{Q}[X]}{Z}$.
  Then $f(s) = \hat{f}(s) = \hat{f}(s - p) + \hat{f}(p) = \hat{f}(s - p) + p \in \algconeX{Z, P}{X}$.
  
  
  ($\supseteq$) Let $p \in \mathbb{Q}[Y]$ be such that $f(p) \in \algconeX{Z, P}{X}$.
  By considering the extended ring homomorphism $\hat{f}$ above,
  we have by Lemma~\ref{lem:eval-ring-hom} that
  $p - f(p) \in \product{\mathbb{Q}[X, Y]}{T}$.
  Since $f(p) \in \algconeX{Z, P}{X}$,
  $p = (p - f(p)) + f(p) \in \product{\mathbb{Q}[X,Y]}{T} + \algconeX{Z,P}{X} \subseteq \algconeX{Z \cup T,P}{X,Y}$.
\end{proofEnd}

\begin{theoremEnd}[restate]{lemma}
  \label{cor:inverse-linear-groebner}
  Let $Z, P \subseteq \mathbb{Q}[X]$ be finite with $Z$ a Gr\"{o}bner basis.
  Let $Y$ be a finite set of variables distinct from $X$.
  Let $g : \linpoly{Y} \to \mathbb{Q}[X]$ be a linear map.
  Then
  \[
  \algconeX{\textit{inverse-hom}(Z,P,\hat{g},Y)}{Y} \cap \linpoly{Y} = g^{-1}(\algconeX{Z,P}{X}),
  \]
  where $\hat{g} : \mathbb{Q}[Y] \to \mathbb{Q}[X]$ is the unique ring
  homomorphism satisfying $\hat{g} |_{\linpoly{Y}} = g$.
\end{theoremEnd}
\begin{proofEnd}
  Let $T = \set{y - \hat{g}(y): y \in Y}$.
  Let $C = \algconeX{Z, P}{X}$.
  For $\subseteq$,
  $\algconeX{T \cup Z, P}{X, Y} \cap \linpoly{Y} \subseteq \algconeX{T \cup Z, P}{X, Y} \cap \mathbb{Q}[Y] = \hat{g}^{-1}(C)$,
  so $g(\algconeX{T \cup Z, P}{X, Y} \cap \linpoly{Y}) \subseteq C$.
  For $\supseteq$, note that
  $g^{-1}(C) \subseteq \hat{g}^{-1}(C) \subseteq \algconeX{T \cup Z, P}{X, Y}$,
  and $g^{-1}(\algconeX{Z, P}{X}) \subseteq \linpoly{Y}$ by definition.
\end{proofEnd}

Note that $\linpoly{Y}$ is an algebraic cone in $\mathbb{Q}[Y]$,
so we can do the intersection using the procedure of Section~\ref{sec:intersection}.
Precisely, let $Y \subseteq X$ be sets of variables.
Define
\[
\textit{intersect-subspace}(Z, P, Y)
= \textit{second}(\textit{intersect}(Z, P, \emptyset, P_{\linpoly{Y}})),
\]
where $P_{\linpoly{Y}} = Y \cup -Y \cup \set{1, -1}$,
and $\textit{second}$ returns the second component of the pair.
Note that the first component is always $\emptyset$ (or equivalently, $\set{0}$) since it is a Gr\"{o}bner basis for the ideal $\qideal{Z} \cap \qideal{\emptyset} = \set{0}$.

Let $\polyhedralcut(P, Y)$ be a procedure that computes
the cutting plane closure of a polyhedral cone generated by $P$, 
with respect to $\product{\mathbb{Z}}{Y}$.
That is, 
$\product{\Qplus}{\polyhedralcut(P, Y)} = \cp_{\product{\mathbb{Z}}{Y}}(\product{\Qplus}{P})$.
This may be done by e.g., the iterated Gomory-Chv\'{a}tal closure.

We may now arrive at an effective implementation of the reduction in Lemma~\ref{lem:cp-ineq-closure}.  Define an operation $\cut(Z, P, B)$ by:
\[
\cut(Z, P, B) \defeq \textit{substitute}(\hat{f}, \polyhedralcut(\textit{intersect-subspace}(\textit{inverse-hom}(Z, P, \hat{f}, Y), Y))),
\]
where $Y = \set{y_1, \ldots, y_n}$ is a set of fresh variables corresponding
to $B = \set{v_1, \ldots, v_n}$ generating a point lattice,
$\hat{f}: \mathbb{Q}[Y] \to \mathbb{Q}[X]$ is the ring homomorphism defined by
$\hat{f}(y_i) = v_i$,
and
$\textit{substitute}(\hat{f}, P')$ applies $\hat{f}$ to each
polynomial in $P'$.
By linearity of $\hat{f}$, $\textit{substitute}(\hat{f}, P')$ computes the image
of the polyhedral cone generated by $P'$.

This procedure implements the key step in computing cutting plane closure
(Lemma~\ref{lem:cp-ineq-closure}).

\begin{theoremEnd}[restate]{lemma}
  \label{lem:cut-correct}
  Let $Z, P, B \subseteq \mathbb{Q}[X]$ with $Z$ a Gr\"{o}bner basis.
  Let $Y$ and $\hat{f}$ be as defined in $\cut$,
  and let $f = \hat{f} |_{\linpoly{Y}}$ be the linear map obtained by restricting
  $\hat{f}$.
  Then
  \[
  \product{\Qplus}{\cut(Z, P, B)} = f(\cp_{\, \product{\mathbb{Z}}{Y}}(f^{-1}(\algconeX{Z, P}{X}))).
  \]
\end{theoremEnd}
\begin{proofEnd}
  By Corollary~\ref{cor:inverse-linear-groebner}, the definition of \textit{intersect-subspace},
  and Theorem~\ref{thm:intersection},
  \[
  \product{\Qplus}{\textit{intersect-subspace}(\textit{inverse-hom}(Z, P, \hat{f}, Y), Y)} = f^{-1}(\algconeX{Z, P}{X}))).
  \]

  Let $P_1 = \textit{intersect-subspace}(\textit{inverse-hom}(Z, P, \hat{f}, Y), Y)$.
  The $\polyhedralcut$ procedure computes cutting plane closure, so
  $\product{\Qplus}{\polyhedralcut(P_1)} = \cp_{\, \product{\mathbb{Z}}{Y}}(\product{\Qplus}{P_1})$.

  Let $P_2 = \polyhedralcut(P_1)$.
  Then
  \[
  \begin{array}{llr}
    \product{\Qplus}{\textit{substitute}(\hat{f}, P_2)}
    &= \product{\Qplus}{\set{\hat{f}(p) : p \in P_2}} & \\
    &= \product{\Qplus}{\set{f(p) : p \in P_2}} & \text{ since $Z_3, P_3 \subseteq \linpoly{Y}$}\\
    &= f(\product{\Qplus}{P_2}) & \text{ by linearity of $f$}.
  \end{array}
  \]
  Chaining these parts yields the claim.
\end{proofEnd}

\begin{mexample}
  Let $Z = \emptyset$, $P = \set{x_1 - 2x_2 + 1, x_1 + 2x_2, -x_1}$
  and $B = \set{2x_1, 2x_2}$.
  To compute $\cut(Z, P, B)$,
  we first do $\textit{inverse-hom}(Z, P, \hat{f}, \set{y_1, y_2})$
  with $\hat{f}(y_1) = 2x_1$ and $\hat{f}(y_2) = 2x_2$.

  Per the definition of \textit{inverse-hom} for \cut,
  we do $\project_Y(\set{y_1 - 2x_1, y_2 - 2x_2}, P)$ with respect to
  an elimination order $\preceq_Y$.
  The set $\set{y_1 - 2x_1, y_2 - 2x_2}$ is already a Gr\"{o}bner basis
  with respect to such an order,
  so we may use it to reduce $P$ to get
  $P' = \set{ \frac{1}{2} y_1 - y_2 + 1, \frac{1}{2} y_1 + y_2, -\frac{1}{2}y_1}$.
  These polynomials are all in $\mathbb{Q}[y_1, y_2]$.
  Since $\set{y_1 - 2x_1, y_2 - 2x_2} \cap \mathbb{Q}[y_1, y_2] = \emptyset$,
  $\project_Y$, hence \textit{inverse-hom}, returns $\tuple{(\emptyset , P')}$.

  Next, we do $\textit{intersect-subspace}(\emptyset, P', \set{y_1, y_2})$.
  The polynomials are already linear in $y_1$ and $y_2$, so
  \textit{intersect-subspace} returns $P'$.

  For \polyhedralcut, we consider the polyhedron
  \[
  Q = \set{ (y_1, y_2) \in \mathbb{Q}^2:
    \frac{1}{2} y_1 - y_2 + 1 \geq 0,
    \frac{1}{2} y_1 + y_2 \geq 0,
    -\frac{1}{2}y_1 \geq 0}
  \]
  defined by $P'$.
  This polyhedron has vertices $(0, 0)$, $(0, 1)$ and $(-1, \frac{1}{2})$,
  so its integer hull is defined by
  $y_1 = 0$, $y_2 \geq 0$, and $y_2 \leq 1$.
  Hence, the cutting plane closure is
  $\polyhedralcut(P') = \set{1, y_1, -y_1, y_2, -y_2 + 1}$,
  where $1$ is present because $1 \geq 0$ is always a valid inequality.

  Finally, applying $\textit{substitute}$ gives
  $\cut(Z, P, B) = \set{1, 2x_1, -2x_1, 2x_2, -2x_2 + 1}$.
  By Lemma~\ref{lem:cp-ineq-closure},
  \begin{align*}
  \cp_{\product{\mathbb{Z}}{B}}(\regcone(Z, P))
  &= \regcone(Z, P) + \product{\Qplus}{\set{1, 2x_1, -2x_1, 2x_2, -2x_2 + 1}} \\
  &= \product{\mathbb{Q}}{\set{x_1}} + \product{\Qplus}{1, x_2, -2x_2 + 1}. \qedhere
  \end{align*}
\end{mexample}

Next, we show that minimal models can be constructed by
iterating the regular
closure (Algorithm~\ref{alg:saturation}) and
cutting plane closure separately until fixed point.

\begin{lemma}
  \label{lem:saturate-cp-closure}
  For any $C \subseteq \mathbb{Q}[X]$,
  let $R(C)$ be the least regular cone that contains $C$.
  Define $T_L(C) = R(\cp_{L}(C))$.
  For any regular cone $C$ and polynomial lattice $L$,
  there exists $n \in \Zplus$ such that $T_L^n(C)$ is the least regular
  cone that contains $C$ and is closed under cutting planes
  with respect to $L$.
\end{lemma}
\begin{proof}
  Let $C_0 = C$ and $C_{i+1} = R(\cp_L(C_i)) = T_L(C_i)$.
  Since each $C_i$ is a regular polynomial cone,
  and since $\mathbb{Q}[X]$ is Noetherian,
  the sequence $(\units{C_i})_i$ stabilizes at some $n$.

  Note that $\units{C_n} \subseteq \units{\cp_L(C_n)} \subseteq \units{R(\cp_L(C_n))} = \units{C_{n+1}} = \units{C_n}$,
  so $\units{\cp_L(C_n)} = \units{C_n}$.
  Hence, $\units{\cp_L(C_n)}$ is an ideal.
  Since $1 \in C$, $1 \in \cp_L(C_n)$ for all $n$,
  so $\cp_L(C_n)$ is regular and $R(\cp_L(C_n)) = \cp_L(C_n)$.
  Then $C_{n+2} = R(\cp_L(C_{n+1})) = R(\cp_L(R(\cp_L(C_n)))) = R(\cp_L(\cp_L(C_n))) = R(\cp_L(C_n)) = C_{n+1}$.
  Thus, $C_0, C_1, \ldots$ stabilizes at $n+1$.

  It is clear that $C = C_0 \subseteq C_1 \subseteq \ldots \subseteq C_{n+1}$,
  so $C_{n+1}$ contains $C$.
  Since $\units{C_0}, \units{C_1}, \ldots$ stabilizes at $n$, $C_{n+1}$ is regular.
  Since $C_0, C_1, \ldots$ stabilizes at $n+1$, $C_{n+1}$ is closed under
  cutting planes with respect to $L$.

  By induction on $i$, each $C_i$ is contained in any cone satisfying
  these properties, so $C_{n+1}$ is the least such cone.
\end{proof}

With these results, $\textit{rcp}(Z, P, B)$
(Algorithm~\ref{alg:rcp-closure}) computes the least
regular cone $C$ that contains $\regcone(Z, P)$ and is closed under
cutting planes with respect to $\units{C} + \product{\mathbb{Z}}{B}$.

\begin{algorithm}[t]
  \SetAlgoLined
  \SetKwFunction{FRcp}{rcp}

  \SetKwFunction{FSaturate}{saturate}
  \SetKwFunction{FCut}{cut}

  \Fn{\FRcp{$Z, P, B$}}{
    \Input{$\tuple{Z, P, B}$, with $\tuple{Z, P}$ reduced.}
  \Output{$\tuple{Z', P'}$ such that $\tuple{Z', P'}$ is reduced and $\regcone(Z', P')$
    is the least regular cone $C$ that contains $\regcone(Z, P)$ and is
    closed under cutting planes with respect to $\units{C} + \product{\mathbb{Z}}{B}$.}
  $\tuple{Z', P'} \gets \tuple{Z, P}$\;
  $B' \gets$ a basis for $\product{\mathbb{Z}}{B}$\;
  \Repeat{$\regcone(Z', P') = \regcone(Z_0, P_0)$}{ \label{alg-line:while-loop}
    $\tuple{Z_0, P_0} \gets \tuple{Z', P'}$\;
    $P_1 \gets \FCut(Z_0, P_0, B')$\;
    $\tuple{Z', P'} \gets \FSaturate( Z_0 \cup -Z_0 \cup P_0 \cup P_1)$\;
    $B' \gets$ basis for $\product{\mathbb{Z}}{\set{\red_{Z'}(b): b \in B'}}$\;
  }
  \Return $\tuple{Z', P'}$\;
  }
  \caption{Regular cutting plane closure of a cone \label{alg:rcp-closure}}
\end{algorithm}

\begin{theoremEnd}[restate]{theorem}[Regular cutting plane closure]
  \label{thm:rcp}
  Given $Z, P, B$ with $\tuple{Z, P}$ reduced,
  $\textit{rcp}(Z, P, B)$ computes the least
  regular cone $C$ that contains $\regcone(Z, P)$ and is closed under
  cutting planes with respect to $\units{C} + \product{\mathbb{Z}}{B}$.
\end{theoremEnd}
\begin{proofEnd}
  Let $B_0^{(i)}$ be $B'$ at the start of the $i$-th iteration.
  Let $Z_0^{(i)}$ and $C^{(i)}$ be
  $Z_0$ and $\regcone(Z_0, P_0)$ in the $i$-th iteration
  respectively.
  By Theorem~\ref{thm:saturation},
  Lemma~\ref{lem:cp-ineq-closure} and Lemma~\ref{lem:cut-correct},
  $C^{(i+1)} = R(\cp_{\, \product{\mathbb{Z}}{B_0^{(i)}}}(C^{(i)}))$,
  where $R(C)$ is the least regular cone containing $C$.
  We first claim that for all $0 \leq i \leq j$,
  \begin{align*}
    \cp_{\, \product{\mathbb{Z}}{B_0^{(i)}}}(C^{(j)}) = \cp_{\, \product{\mathbb{Z}}{B}}(C^{(j)}).
  \end{align*}
  Then in particular, $C^{(i+1)} = R(\cp_{\, \product{\mathbb{Z}}{B}}(C^{(i)}))$.
  By Lemma~\ref{lem:saturate-cp-closure},
  the algorithm terminates with $\regcone(Z', P')$ as the least regular cone
  $C$ that contains $\regcone(Z, P)$
  and is closed under cutting planes with respect to $\product{\mathbb{Z}}{B}$.
  By Lemma~\ref{lem:cp-ineq-invariance},
  it is closed under cutting planes with respect to
  $\units{\regcone(Z', P')} + \product{\mathbb{Z}}{B}$.
  If $C'$ is any regular cone containing $\regcone(Z, P)$ and is closed
  under cutting planes with respect to $\units{C'} + \product{\mathbb{Z}}{B}$,
  it is in particular closed under cutting planes with respect to $\product{\mathbb{Z}}{B}$,
  so $\regcone(Z', P') \subseteq C'$. This proves the theorem.

  We prove the claim by induction on $i$.
  The base case follows by noting that
  $\product{\mathbb{Z}}{B} = \product{\mathbb{Z}}{B_0^{(0)}}$.
  Now assume the induction hypothesis, i.e.,
  for all $j$ such that $0 \leq i \leq j$,
  $\cp_{\, \product{\mathbb{Z}}{B_0^{(i)}}}(C^{(j)}) = \cp_{\, \product{\mathbb{Z}}{B}}(C^{(j)})$.
  We wish to show that if $i + 1 \leq j$,
  $\cp_{\, \product{\mathbb{Z}}{B_0^{(i+1)}}}(C^{(j)}) = \cp_{\, \product{\mathbb{Z}}{B}}(C^{(j)})$.
  Each $b^{(i+1)} \in B_0^{(i+1)}$ is obtained by reducing $B_0^{(i)}$ using
  $Z_0^{(i+1)}$ and computing a basis, so
  $b^{(i+1)} = \sum_k a_k (b_k^{(i)} - z_k^{(i+1)})$, where each
  $a_k \in \mathbb{Z}$, $b_k^{(i)} \in B_0^{(i)}$, and $z_{k}^{(i+1)} \in \qideal{Z_0^{(i+1)}}$.
  Since $i + 1 \leq j$ and $\units{C^{(i+1)}} \subseteq \units{C^{(j)}}$,
  $\product{\mathbb{Z}}{B_0^{(i+1)}} \subseteq \units{C^{(i+1)}} + \product{\mathbb{Z}}{B_0^{(i)}} \subseteq \units{C^{(j)}} + \product{\mathbb{Z}}{B_0^{(i)}}$.
  Symmetrically,
  each $b^{(i)} = z^{(i+1)} + \sum_k a_k b_k^{(i+1)}$,
  where each $a_k \in \mathbb{Z}$, $b_k^{(i+1)} \in B_0^{(i+1)}$, and $z^{(i+1)} \in Z_0^{(i+1)}$.
  Hence, $\product{\mathbb{Z}}{B_0^{(i)}} \subseteq \units{C^{(j)}} + \product{\mathbb{Z}}{B_0^{(i+1)}}$.
  By Lemma~\ref{lem:cp-ineq-invariance},
  $\cp_{\, \product{\mathbb{Z}}{B_0^{(i+1)}}}(C^{(j)}) = \cp_{\, \product{\mathbb{Z}}{B_0^{(i)}}}(C^{(j)})$.
  By the induction hypothesis, the claim follows.  
\end{proofEnd}

\paragraph{Satisfiability and consequence-finding modulo $\ThZ$}

Summarizing, we have the following decision procedure for satisfiability of
conjunctive $\sigma_{or}^Z(X)$-formulas modulo $\ThZ$: given a formula
  \[
  F = \left(\bigwedge_{p \in P} 0 \leq p \right) \land
  \left( \bigwedge_{q \in Q} \lnot (0 \leq q) \right) \land
  \left( \bigwedge_{r \in R} \lnot(0 = r) \right) \land
  \left( \bigwedge_{s \in S} \Int(s) \right) \land
  \left( \bigwedge_{t \in T} \lnot \Int(t) \right)
  \]
  First compute a representation $\tuple{Z',P'}$ of the
  least regular cone containing $P$ and closed under cutting planes with
  respect to $Z' + \product{\mathbb{Z}}{S \cup \set{1}}$ using
  Algorithm~\ref{alg:rcp-closure}, and compute a basis $B'$ for the point
  lattice $\product{\mathbb{Z}}{\red_{Z'}(S) \cup \set{1}}$.  If
  $\red_{Z'}(1) = 0$, then $F$ is unsatisfiable
  ($\regcone(Z',P')$ is inconsistent).  Otherwise, check
  whether $\mathfrak{A} = \ConeModel{\regcone(Z',P'),\qideal{Z'}+\product{\mathbb{Z}}{B'}}$ satisfies $F$ by
  testing whether there is some $q \in Q$ with $\red_{Z'}(q) \in
  \product{\Qplus}{P'}$ (Lemma~\ref{lem:membership}), or some $r \in
  R$ with $\red_{Z'}(r) = 0$,
  or some $t \in T$ with $\red_{Z'}(r) \in \product{\mathbb{Z}}{B'}$ (Lemma~\ref{lem:lattice-membership});
  if such a $q$, $r$, or $t$ exists, then $F$
  is unsatisfiable (Theorem~\ref{thm:sat-correctness}), otherwise,
  $\mathfrak{A}$ satisfies $F$.


  Consequence-finding modulo $\ThZ$ operates in the same way as Algorithm~\ref{alg:consequence}.  
  The only difference is that
  $\texttt{get-model}$
returns a triple $\tuple{Z', P', B'}$
instead of pair $\tuple{Z', P'}$---but the point lattice basis $B'$ is simply ignored by the rest of the algorithm.


%% file: invgen.tex
\newcommand{\lininv}{\textit{LinInv}}

\section{Invariant Generation modulo $\ThZ$} \label{sec:invgen}

In this section, we give an application of consequence-finding modulo $\ThZ$ to the problem of
computing loop invariants. Intuitively, we want to compute bounds on linear terms whose increment per iteration of the loop is bounded by a (possibly nonlinear) quantity that remains fixed throughout the execution of the loop. The key property of our loop-invariant generation procedure is that
it is \textit{monotone}: if more information about the program's behavior is given to the procedure, then it may only compute invariants that are more precise (modulo $\ThZ$).  Monotonicity is achieved
due to the fact that we can find and manipulate the set of \textit{all} implied inequalities of a formula modulo $\ThZ$.  Section~\ref{sec:linear-inv} defines and motivates the particular variant of loop invariant generation that we consider, and the invariant generation procedure is presented in Section~\ref{sec:lirr-star}.

\subsection{Approximate reflexive transitive closure} 
\label{sec:linear-inv}

In this section, we will address the problem of over-approximating the reflexive transitive closure of a transition formula (a notion that we will make precise in the following).  Conceptually, we might think of a transition formula as a logical summary of the action of a program through one iteration of the body of some loop, and an approximation of its reflexive transitive closure as a summary for (an unbounded number of iterations of) the loop.  Using the algebraic program analysis framework, a procedure for over-approximating the transitive closure of a transition formula can be ``lifted'' to a whole-program analysis \cite{FMCAD:FK2015,kincaid2021algebraic,zhu2021termination}, allowing us to focus our attention on this relatively simple logical sub-problem.

Fix a finite set of symbols $X$ (corresponding to the variables in a program of interest)
and a set of ``primed copies'' $X' = \set{ x' : x \in X}$.
Define a \textbf{transition formula} $F$ to be a $\sigma_{or}^Z(X \cup X')$ formula, which represents
a relation over program states, where the unprimed variables correspond to the pre-state and the primed variables correspond to the post-state.  For any pair of transition
formulas $F$ and $G$, define their sequential composition as
$ F \circ G \defeq \exists X''. F[X' \mapsto X''] \land G[X \mapsto X''] $
where $X'' = \set{ x'' : x \in X }$ is a set of variable symbols disjoint from $X$ and $X'$
(representing the intermediate state between a computation of $F$ and a computation of $G$).
Define the $t$-fold composition of a transition formula $F$ as $ F^t\defeq F \circ \dotsi \circ F$ ($t$ times).
The problem of over-approximating reflexive transitive closure is, \textit{given a transition formula $F$, find a transition formula $F^\star$ such that $F^t \models_{\ThZ} F^\star$ for all $t \in \mathbb{N}$}.

Our approach to this problem is based on the one in \cite{TCS:ACI2010,FMCAD:FK2015}.  Given a transition formula $F$, this approach computes an over-approximation of its reflexive transitive closure $F^\star$ in two steps.
In the first step, it extracts a system of recurrences $F \models_{\LRA} \bigwedge_{i=1}^n r_i' \leq r_i + a_i$
where each $r_i$ denotes a linear term over $X$, each $r_i'$ denotes a corresponding linear term over $X'$ and each $a_i$ a rational
(for instance $(x'+y') \leq (x+y) + 2$, indicating that the sum of $x$ and $y$ increase by at most two).
In the second step, it computes
their closed forms $F^\star \defeq \exists t. t \geq 0 \land \bigwedge_{i=1}^n r_i' \leq r_i + ta_i$.
 In the following, we generalize this strategy to the non-linear setting by considering recurrences where $a_i$ is not a rational number but rather an \textit{invariant} of the loop -- a polynomial in $X$ that does not change under the action of $F$.


\subsection{An Operator for Approximate Transitive Closure} \label{sec:lirr-star}

This section describes the class of recurrences that we wish to compute
(Section~\ref{sec:recurrent-diff}), a method to compute them
(Section~\ref{sec:compute-rec-diff}), and then defines an approximate
transitive closure operator and shows that it is monotone and (in a sense)
complete modulo $\ThZ$ (Section~\ref{sec:tc}).

\subsubsection{Recurrent differences} \label{sec:recurrent-diff}
Lift the mapping from symbols in $X$ to their primed counterpart in $X'$ into a ring homomorphism of polynomial terms $(-)' : \mathbb{Q}[X] \rightarrow \mathbb{Q}[X']$ (so, e.g., $(xy + 3z)'$ denotes the polynomial $(x'y' + 3z')$).
Define the space of \textit{linear invariant functionals} $\lininv(F)$ of
a formula $F$ to be those linear terms over $X$ that are invariant under the action of $F$:
\[ \lininv(F) \defeq \set{ k \in \product{\mathbb{Q}}{X} : F \models_{\ThZ} k' = k }\ . \]
Define the space of \textit{invariant polynomials}
$\textit{Inv}(F)$ to be the subring of $\mathbb{Q}[X]$ generated by $\lininv(F)$.  Clearly, for any $p \in \textit{Inv}(F)$ we have 
$F \models_{\ThZ} p' = p$ (but the reverse does not hold).
Finally, we define the class of recurrences of interest, the \textit{recurrent differences} of $F$, to be
\[ \textit{rec}(F) \defeq \set{ \tuple{r,a} \in \product{\mathbb{Q}}{X} \times \textit{Inv}(F) : F \models_{\ThZ} r' \leq r + a}\ . \]
Note that for any $\tuple{r,a} \in \textit{rec}(F)$ and any $t \in \mathbb{N}$, we have $F^t \models_{\ThZ} r' \leq r + ta$.

\begin{mexample} \label{ex:star-running}
  Consider the program \cinline{while(*) \{ P \}}, where $P$ is the program
  \begin{lstlisting}[style=base,mathescape,belowskip=-0.8\baselineskip]
    if (*) { w = w + 1; } else { w = w + 2; }
    x = x + z ; y = y + z ; z = (x - y)(x - y);
  \end{lstlisting}
  The transition formula for $P$ is
  \[
  F = (w' = w + 1 \lor w' = w + 2) \land x' = x + z \land y' = y + z \land z' = z + (x' - y')^2.
  \]
  Notice that the linear polynomial $x - y$ is an invariant of the loop;
  $\lininv(F)$ is the linear space spanned by $(x-y)$.
  The polynomial ${(x - y)}^2$ is thus an invariant polynomial in the subring
  $\textit{Inv}(F)$.
  The recurrent differences
  $\textit{rec}(F)$ includes (but is not limited to) $\tuple{-w, -1}$ ($w$ increases by at least 1), $\tuple{w,2}$ ($w$ increases by at most 2), as well as $\tuple{z, (x-y)^2}$
  and $\tuple{-z, -(x-y)^2}$ ($z$ is always set to $z + (x-y)^2$ at each iteration).
\end{mexample}

\begin{mexample}
  Consider the program \cinline{while (*) \{ w = -w; \}}.
  The transition formula $F$ of the loop body is $w' = -w$.
  The polynomial $w^2$ is an invariant polynomial of the loop:
  $F \models_{\ThZ} {(w')}^2 = w^2$.
  However, $F$ has no \textit{linear} invariants,
  so $\lininv(F) = \set{0}$ and $\textit{Inv}(F) = \set{0}$.
  Such invariants are not considered in this paper.
\end{mexample}

\subsubsection{Computing recurrent differences} \label{sec:compute-rec-diff}

We proceed in three steps, showing how to represent and compute $\lininv(F)$, $\textit{Inv}(F)$, and finally $\textit{rec}(F)$.

Since $\lininv(F)$ is a linear space it can be represented by a basis, which we compute as follows.
Let $D \defeq \set{ d_x : x \in X }$ denote a set of variables distinct from those in $X,X'$.
Define a ring homomorphism $\delta : \mathbb{Q}[D] \rightarrow \mathbb{Q}[X \cup X']$
by $\delta(d_x) \defeq x - x'$ and a
second ring homomorphism $\textit{pre} : \mathbb{Q}[D] \rightarrow \mathbb{Q}[X]$
by $\textit{pre}(d_x) = x$.
Then $\lininv(F) = \textit{pre}(\delta^{-1}(\units{\Cn{F}}) \cap \product{\mathbb{Q}}{D})$:
$F \models_{\ThZ} k' = k$ iff $F \models_{\ThZ} k' - k = 0$
iff $k' - k$ is a unit in the consequence cone
of $F$;
since $k$ is linear, $k' - k$ is a linear combination of the difference variables $D$.
A basis for $\lininv(F)$ can be computed using the primitives we have developed in the preceding
(consequence-finding, cone intersection, inverse image, and image).

Let $\set{a_1,\dots,a_n}$ be a basis for $\lininv(F)$.  The subring $\textit{Inv}(F)$ of $\mathbb{Q}[X]$ can be represented as the elements of the polynomial ring $\mathbb{Q}[K]$, where $K = \set{ k_1, \dots, k_n }$ is a set of fresh variables, one for each basis element of $\lininv(F)$.  Let $\textit{inv}$ denote the (injective) ring homomorphism $\mathbb{Q}[K] \rightarrow \mathbb{Q}[X]$ that sends $k_i$ to $a_i$ for each $i$.  Then the image of $\textit{inv}$ is precisely $\textit{Inv}(F)$ (i.e., $\mathbb{Q}[K]$ and $\textit{Inv}(F)$ are isomorphic).


Finally, we show how to represent and compute the set of recurrences $\textit{rec}(F)$.  Each element of $\textit{rec}(F)$ is a pair $\tuple{r,a}$ consisting of a linear term $r \in \product{\mathbb{Q}}{X}$ and a polynomial $a \in \textit{Inv}(F)$.  As a technical convenience, we can represent such a pair as a polynomial in $\product{\mathbb{Q}}{D} + \mathbb{Q}[K]$.  Since $D$ and $K$ are disjoint, there exist (unique) linear maps $\pi_D : \product{\mathbb{Q}}{D} + \mathbb{Q}[K] \rightarrow \product{\mathbb{Q}}{D}$
and $\pi_K : \product{\mathbb{Q}}{D} + \mathbb{Q}[K] \rightarrow \mathbb{Q}[K]$ such that for all $p \in \product{\mathbb{Q}}{D} + \mathbb{Q}[K]$ we have $p = \pi_D(p) + \pi_K(p)$.  Then we may define a bijection
$\textit{rep} : \product{\mathbb{Q}}{D} + \mathbb{Q}[K] \rightarrow \product{\mathbb{Q}}{X} \times \textit{Inv}(F)$ by $\textit{rep}(p) = (\textit{pre}(\pi_D(p)), \textit{inv}(\pi_K(p)))$.
Define $\widetilde{\textit{rec}}(F)$ to be the inverse image of $\textit{rec}(F)$ under $\textit{rep}$ (i.e., an exact representation of $\textit{rec}(F)$ in the space $\product{\mathbb{Q}}{D} + \mathbb{Q}[K]$).

Thus, it suffices to show how to compute $\widetilde{\textit{rec}}(F)$.
Define a ring homomorphism $\delta_{\textit{inv}} : \mathbb{Q}[D, K] \rightarrow \mathbb{Q}[X \cup X']$ by
$\delta_{\textit{inv}}(d_x) = x - x'$ and $\delta_{\textit{inv}}(k_i) = a_i$ (i.e., the common extension of $\delta$ and $\textit{inv}$).
Then we see that
$\widetilde{\textit{rec}}(F) \defeq \delta_{\textit{inv}}^{-1}( \Cn{F} ) \cap (\product{\mathbb{Q}}{D} + \mathbb{Q}[K])$;
that is,
the recurrences of $F$ correspond exactly to the members of $\delta_{\textit{inv}}^{-1}( \Cn{F} )$ that are linear in the $D$ variables.
We illustrate this with an example.

\begin{mexample}\label{ex:star-running-running}
  We continue Example~\ref{ex:star-running}.
  The consequence cone of $F$ is $\Cn{F} = \regcone(Z, P)$, where $Z$ and $P$ are computed as:
  \begin{align*}
    Z &= \set{(w - w' + 1)(w - w' + 2), x - x' + z, y - y' + z, z - z' + {(x- y)}^2 } \\
    P &= \set{-w + w' - 1, w -w' + 2}
  \end{align*}
  Applying $\delta^{-1}$ to the units of this cone gives
  $\product{\mathbb{Q}[D]}{\set{(\delta_w + 1)(\delta_w + 2), \delta_x - \delta_y}}$
  (since $\delta(\delta_x - \delta_y) = (x - x' + z) - (y - y' + z)$).
  Intersecting this with
  $\product{\mathbb{Q}}{D}$ yields $\product{\mathbb{Q}}{\set{\delta_x - \delta_y}}$.
  Hence, $\lininv(F) = \product{\mathbb{Q}}{\set{x - y}}$.

  The subring $\textit{Inv}(F)$ generated by these linear invariants is
  represented by $\mathbb{Q}[k_1]$ (with $\textit{inv}(k_1) = x-y$).
  We have the following corresponding elements in
  $\textit{rec}(F)$ and $\widetilde{\textit{rec}}(F)$:
  \[
  \tuple{-w,-1} \sim -\delta_w -1, \tuple{w,2} \sim \delta_w + 2, \tuple{z,(x-y)^2} \sim \delta_z + k_1^2, \text{ and } \tuple{-z,-(x-y)^2} \sim -\delta_z - k_1^2
  \]
  Notice how $-\delta_w -1$ and $\delta_w + 2$ correspond to $P$,
  and how $\delta_z + k_1^2$ and $-\delta_z - k_1^2$
  correspond to the last polynomial $(z - z') + {(x- y)}^2 \in Z$.
\end{mexample}
We now continue with showing how to compute $\widetilde{\textit{rec}}(F)$.
Since we already saw how to compute inverse images of algebraic cones, it remains only to show that we can compute the intersection of an algebraic cone over $\mathbb{Q}[D,K]$ with $\product{\mathbb{Q}}{D} + \mathbb{Q}[K]$; i.e., the set of polynomials in cone that are linear in the set of variables $D$, but may contain arbitrary monomials in $K$.  This is \textit{nearly} solved by the intersection algorithm in Section~\ref{sec:intersection}, since
$\product{\mathbb{Q}}{D} + \mathbb{Q}[K]$ is the sum of a finitely-generated cone $\product{\mathbb{Q}}{D} = \cone{D \cup -D}$ and an ideal $\mathbb{Q}[K]$; however, $\product{\mathbb{Q}}{D} + \mathbb{Q}[K]$ is not an algebraic cone because $\mathbb{Q}[K]$ is not an ideal \textit{in $\mathbb{Q}[D,K]$}.  However, the essential process behind cone intersection carries over, which yields Algorithm~\ref{alg:lin-restrict}.

\begin{algorithm}[t]
  \SetKwFunction{Fexp}{exp}
  \SetKwFunction{Flin}{lin}
  \Fn{\Flin{$Z, P, D, K$}}{
    \Input{Finite sets of polynomials $Z,P \subseteq \mathbb{Q}[D,K]$ over disjoint variables $D$ and $K$}
    \Output{A pair $\tuple{V,R}$ with
    $\product{\mathbb{Q}[K]}{V} + \cone{R} = \regcone(Z,P) \cap (\mathbb{Q}[K] + \product{\mathbb{Q}}{D})$}
    $\tuple{Z,P} \gets \reducecone{Z,P}{\preceq_{K}}$\;
    \uIf{$\red_Z(1) = 0$}{
        \Return{$\tuple{\set{1}, D \cup -D}$} \tcc*{$\regcone(Z,P) = \mathbb{Q}[D,K]$, so return $\mathbb{Q}[K] + \product{\mathbb{Q}}{D})$}
    }
    $R \gets \polyhedralproject{P}{D \cup [K]}$ \tcc*{Project onto the set of variables $D$ and monomials over $K$}
    $V \gets \emptyset$\;
    \tcc{Add polynomials in $Z \cap \mathbb{Q}[K]$ to $V$; if a polynomial is linear in $d$, add it (and its negation) to $R$}
    \ForEach{$z \in Z$}{
      \uIf{$\leadmon(z) \in \mathbb{Q}[K]$}{
        $V \gets V \cup \set{z}$
      }
      \uElseIf{$\leadmon(z) \in \product{\mathbb{Q}}{D}$}{
        $R \gets R \cup \set{ z, -z }$
      }
    }
    \Return{$\tuple{V, R}$}
  }
  \caption{Linear restriction \label{alg:lin-restrict}}
\end{algorithm}

\begin{theoremEnd}[restate]{lemma}
 \label{lem:lin-restrict}
  Let $D,K$ be disjoint finite sets of variables and $Z,P \subseteq \mathbb{Q}[D,K]$ be finite sets of polynomials over $D$ and $K$.  Let $\tuple{V,R} = \Flin(Z,P,D,K)$.  Then
  \[
    \product{\mathbb{Q}[K]}{V}  + \cone{R} = \regcone(Z,P) \cap (\mathbb{Q}[K] + \product{\mathbb{Q}}{D})
  \]
\end{theoremEnd}
\begin{proofEnd}
  Without loss of generality we may assume that $\tuple{Z, P}$ is reduced w.r.t. $\preceq_K$.
  The $\subseteq$ inclusion is obvious, since by construction 
  $V$ is a subset of  
  $\units{\regcone(Z,P)}$ and $\mathbb{Q}[K]$, and $R$ is a subset of $\regcone(Z,P)$ and $(\mathbb{Q}[K] + \product{\mathbb{Q}}{D})$. 
  
  For the $\supseteq$ direction, suppose $q \in \regcone(Z,P) \cap (\mathbb{Q}[K] + \product{\mathbb{Q}}{D})$. 
  Thus $q \in (\mathbb{Q}[K] + \product{\mathbb{Q}}{D})$ can be written as
 $q = z + p$ for $z \in \product{\mathbb{Q}[D,K]}{Z}$, $p \in \cone{P}$.
 According to \autoref{lem:decomp-to-smaller-monos}, $\leadmon(z) \preceq_K \leadmon(q)$ and
 $\leadmon(p) \preceq_K \leadmon(q)$.
 Thus $p \in \cone{P} \cap (\mathbb{Q}[K] + \product{\mathbb{Q}}{D}) \subseteq \cone{R}$,
 since $\cone{R} = \cone{P} \cap (\mathbb{Q}[K] + \product{\mathbb{Q}}{D})$ by properties
 of cone projection in Section~\ref{sec:subsub-proj-alg-cones}. We thus only need 
 to show that $z \in \cone{R} + \product{\mathbb{Q}[K]}{V}$.
 
 Suppose $z = q_1z_1 + \dots + q_nz_n$ for $q_1, \dots, q_n \in \mathbb{Q}[D, K]$ and 
 $z_1, \dots, z_n \in Z$.
 Then $\leadmon(q_iz_i) \preceq_K \leadmon(z)$ for all $i$.
 Next we proceed based on what $\leadmon(z_i)$ looks like for each $i$.
 
 If $\leadmon(z_i) \in \mathbb{Q}[K]$. In this case $q_i$ must not involve any variable $D$, 
 since otherwise $\leadmon(q_i z_i)$ will be ranked higher than anything in $(\mathbb{Q}[K] + \product{\mathbb{Q}}{D})$ including $\leadmon(z)$ w.r.t. the monomial ordering $\preceq_K$.
 In other words $q_i \in \mathbb{Q}[K]$. 
 This $z_i$ is added to $V$, thus $q_i z_i \in \product{\mathbb{Q}[K]}{V}$.

 If $\leadmon(z_i)$ contains any variable in $D$, then $q_i \in \mathbb{Q}$ for the same
 reason as above. In this case $q_i z_i \in \cone{R}$ since $z_i$ and $-z_i$ are added to $R$.
\end{proofEnd}

\subsubsection{Approximate Transitive Closure} \label{sec:tc}

The core logic of our approximate transitive closure operator appears in
Algorithm~\ref{alg:exp}, which
computes an overapproximation \Fexp($F$,$t$) of
the $t$-fold composition of $F$ (symbolic in $t$) using the methodology
described in the last two sections.  The approximate transitive closure
operator is defined as $F^\star \defeq \exists t. \Int(t) \land t \geq 0 \land
\Fexp(F, t)$.

\begin{mexample}
  Returning to Example~\ref{ex:star-running-running},
  the recurrent differences computed by $\Flin$ are
  $V = \emptyset$
  and $R = \set{-\delta_w - 1, \delta_w + 2, \pm(\delta_z + k_1^2), \pm(\delta_x - \delta_y)}$.
  Hence, Algorithm~\ref{alg:exp} computes
  \[
  \Fexp(F, t) =
  -(w - w') - t \geq 0 \land w - w' + 2t \geq 0
  \land z - z' + t {(x - y)}^2 = 0
  \land (x - x') - (y - y') = 0.
  \]
  That is, the loop may be summarized as having the effect
  \[
  w + t \leq w' \leq w + 2t \land z' = z + t{(x - y)}^2 \land x - y = x' - y'.
  \]
\end{mexample}

In the following, we prove that Algorithm~\ref{alg:exp}
is sound (Theorem~\ref{thm:star-soundness}),
complete (Lemma~\ref{lem:star-completeness}), and monotone
(Theorem~\ref{thm:star-monotone}).


\begin{algorithm}[t]
  \SetKwFunction{Fsymbinv}{inv-linear-funcs}
  \SetKwFunction{Frecurrence}{recurrence}
  \SetKwFunction{Fexp}{exp}
  \Fn{\Fexp{$F,t$}}{
    $\set{a_1,\dots,a_n} \gets$ basis for $\lininv(F)$\;
    $\tuple{Z,P} \gets \textit{consequence}(F)$\;
    $D = \set{ d_x : x \in X} \gets $ set of $|X|$ fresh variables\;
    $K = \set{k_1,\dots,k_n} \gets$ set of $n$ fresh variables\;
    Let $\delta_{\textit{inv}} : \mathbb{Q}[K \cup D] \rightarrow \mathbb{Q}[X \cup X']$ be the homomorphism mapping $d_x \mapsto x$ and $k_i \mapsto a_i$\;
    $\tuple{Z',P'} \gets \textit{inverse-hom}(Z,P,\delta_{inv},D \cup K)$\;
    $\tuple{V,R} \gets \textit{lin}(Z',P',D,K)$\;
    Define $\textit{cf}$ to be the map that sends $p \mapsto \delta(\pi_{D}(p)) + t \cdot \textit{inv}(\pi_K(p))$\;
    \Return{$\left(\bigwedge_{v \in V} 0 = \textit{cf}(v)\right) \land \left( \bigwedge_{r \in R} 0 \leq \textit{cf}(r)\right)$}
  }
  \caption{Approximation of $t$-fold composition of $F$  \label{alg:exp}}
\end{algorithm}

\begin{theorem}[Soundness] \label{thm:star-soundness}
  Given a transition formula $F$,  $F^t \models_{\ThZ} F^\star$ for all $t \in \mathbb{N}$.
\end{theorem}
\begin{proof}
  Fix any integer $m$. We just need to show that $F^{m} \models_{\ThZ}
  \Fexp(F, m)$.  $\Fexp(F, m)$ is a conjunction of predicates of the form $0 =
  \textit{cf}(v)$ for $v \in V$, and $0 \leq \textit{cf}(r)$ for $r \in R$
  (with $V, R$ as in \autoref{alg:exp}).  It is sufficient to show that $F^m$
  entails each such predicate.

  Suppose $v \in V$.  Then $F \models_{\ThZ} 0 = \textit{inv}(v)$, and so $F^m \models_{\ThZ}  \textit{inv}(v)$ and therefore
  $F^m \models_{\ThZ}  \textit{cf}(v)$.

  Suppose $r \in R$.  Then $F \models_{\ThZ} 0 \leq \delta(\pi_D(r)) + \textit{inv}(\pi_K(r))$.  Since we also have $F \models_{\ThZ} \textit{inv}(\pi_K(r)) = \textit{inv}(\pi_K(r))'$, we have $F^m \models_{\ThZ} 0 \leq \delta(\pi_D(r)) + m\textit{inv}(\pi_K(r))$, and therefore $F^m \models_{\ThZ} \textit{cf}(r)$.
\end{proof}

\begin{lemma}[Completeness] \label{lem:star-completeness}
  Given any transition formula $F$,
  we have $\textit{exp}(F,t) \models_{\ThZ} r' \leq r + ta$ 
  for all $\tuple{r,a} \in \textit{rec}(F)$.
\end{lemma}
\begin{proof}
  We have established a bijection between $\textit{rec}(F)$ and $\widetilde{\textit{rec}}(F)$, and that
  $\widetilde{\textit{rec}}(F) = \product{\mathbb{Q}[K]}{V} + \cone{R}$, where $\tuple{V,R}$ is defined on line 8 of \autoref{alg:exp}.
  Suppose $\tuple{r,a} \in \textit{rec}(F)$.  Then there is some $p \in \widetilde{\textit{rec}}(F)$ such that 
  $\textit{rep}(p) = \tuple{r,a}$.
  Let $\textit{cf}$ be the map $p \mapsto \delta(\pi_D(p)) + t \cdot \textit{inv}(\pi_K(p))$, where 
  $t$ is the symbolic loop counter.
  Then $r - r' + ta = \textit{cf}(p)$.
  
  Since
  $p \in \product{\mathbb{Q}[K]}{V} + \cone{R} $, we have
  \[
  p = \sum q_i v_i + \sum \lambda_j r_j \ , \quad q_i \in \mathbb{Q}[K], v_i \in V, \lambda_j \in \Qplus, r_j \in R
  \]
  Since $\textit{cf}$ is linear, we can write 
   \[
   \textit{cf}(p) = \sum \textit{cf}(q_i v_i) + \sum \lambda_j \textit{cf}(r_j) \ , \quad q_i \in \mathbb{Q}[K], v_i \in V, \lambda_j \in \Qplus, r_j \in R
  \]
  
  We have $\textit{cf}(q_iv_i) = t \cdot \textit{inv}(q_iv_i)
  = t \cdot \textit{inv}(q_i)\textit{inv}(v_i) = \textit{inv}(q_i)\textit{cf}(v_i) \in \units{\Cn{\textit{exp}(F,t)}}$,
  since $\textit{inv}$ is a ring homomorphism and
  $\textit{cf}(v_i) \in \units{\Cn{\textit{exp}(F,t)}}$ by construction.
  Also for each $r_i \in R$, we have $\textit{cf}(r_i) \in \Cn{\textit{exp}(F,t)}$.
  Hence $r - r' + ta = \textit{cf}(p) \in \Cn{\textit{exp}(F,t)}$.
\end{proof}

\begin{theorem}[Monotonicity]
  If $F \models_{\ThZ} G$, then $F^\star \models_{\ThZ} G^\star$.
  \label{thm:star-monotone}
\end{theorem}
\begin{proof}
  Since if $F \models_{\ThZ} G$ we must have $\textit{rec}(F) \supseteq \textit{rec}(G)$,
  so by Lemma~\ref{lem:star-completeness} we have
  $\Cn{\textit{exp}(F,t)} \supseteq \Cn{\textit{exp}(G,t)}$.  Since $\textit{exp}(G,t)$ is conjunction of inequalities, this means that  $\textit{exp}(F,t)$ entails each conjunct,
  and thus $\textit{exp}(F,t) \models \textit{exp}(G,t)$.  It follows that
  \[ F^\star = \exists t. \Int(t) \land t \geq 0 \land \textit{exp}(F,t) \models_{\ThZ} \exists t. \Int(t) \land t \geq 0 \land \textit{exp}(G,t) = G^\star\ . \qedhere\]
\end{proof}


%% file: evaluation.tex
\section{Evaluation} \label{sec:evaluation}

We implemented the SMT solvers for the theory of linear real rings (Section~\ref{sec:rational}) and linear integer rings (Section~\ref{sec:integer}),
and the monotone procedure for loop invariant generation
as detailed in Section~\ref{sec:invgen}.
This section evaluates the practical performance of these procedures.

\nic{One might be curious how much precision we have lost when
using the weak theories as compared to the standard theories of real
and integer arithmetic for SMT tasks (Subsection~\ref{sec:smt-benchmark-perf}).}


\nic{Consequence-finding certainly is
interesting from a theoretical point of view, since it computes
a representation of everything that is implied by the theory $\ThZ$.
We are also interested in finding out about its practical value --- whether it indeed extracts
useful consequences that would be useful for client applications
like invariant generation (Subsection~\ref{sec:invgen-benchmark-perf}).}

We consider two questions.
Firstly, $\ThQ$ and $\ThZ$ are weaker than standard arithmetic and may admit
non-standard models; thus, how do they compare against standard theories for
SMT tasks (Subsection~\ref{sec:smt-benchmark-perf})?
Secondly, finding strongest consequences has great theoretical value,
but does that translate into practice for client applications like invariant
generation (Subsection~\ref{sec:invgen-benchmark-perf})?



\subsection{Experimental Setup}

\newcommand{\solvername}{Chilon}

\paragraph{Implementation}
We implemented an SMT solver for $\ThQ$ and $\ThZ$, which we call \solvername{}\footnote{The famous ancient Greek proverb ``less is more'', is attributed to Chilon of Sparta, one of the Seven Sages of Greece.}.
Our implementation relies on \texttt{Z3} (as a SAT solver)~\cite{de-moura-bjorner-2008}
\texttt{Apron}/\texttt{NewPolka}~\cite{jeannet-antoine-2009}
for polyhedral operations,
\texttt{FLINT}~\cite{flint}
for lattice operations,
and \texttt{Normaliz}~\cite{Normaliz}
for integer hull computation\nic{Technically it's now only Hilbert basis computation, with GC done in Duet; should we say this?}\zak{Really, I'd like an experiment that demonstrates the improvement  We should include it in the final version, but let's leave it for now}.
Based on \solvername{}, we implemented the consequence-finding
procedure as well as the approximate transitive closure operator for invariant
generation. We have also integrated the invariant
generation procedure into a static analysis framework that facilitates
evaluation on software verification benchmarks, which is
used as a proxy to evaluate the quality of the generated invariants.

\paragraph{Environment} 
We ran all experiments on an Oracle VirtualBox virtual machine with 
Lubuntu 22.04 LTS (Linux kernel version 5.15 LTS), with
a two-core Intel Core i5-5575R CPU @ 2.80 GHz and 4 GB
of RAM. All tools were run with the \texttt{benchexec} tool under 
a time limit of 2 minutes on all benchmarks.

\paragraph{SMT benchmarks and solvers} 
SMT tasks are taken from the 
quantifier-free nonlinear real arithmetic and nonlinear integer arithmetic
(\texttt{QF_NRA} and \texttt{QF_NIA}) divisions of SMT-COMP 2021\footnote{https://smt-comp.github.io/2021/benchmarks.html}.
For each division, we randomly draw the same number of tasks
from each benchmark suite, resulting in about $100$ tasks.
We compare the performance of \solvername{} on these tasks against
other solvers for standard theories of arithmetic, including
Z3 4.8.14, MathSAT 5.6.7, CVC4 1.8, and Yices 2.6.4. 

\paragraph{Invariant generation benchmarks and tools}
Invariant generation tasks are the safe, integer-only\footnote{That is, the error location is unreachable, and all variables are integer-typed.} benchmarks from
the \texttt{c/ReachSafety-Loops} subcategory of
SV-COMP 2022\footnote{https://sv-comp.sosy-lab.org/2022/}. 
Selected invariant generation tasks from folder \texttt{nla-digbench} and \texttt{nla-digbench-scaling}
constitute the \texttt{nonlinear} benchmark suite, while all other tasks
form the \texttt{linear} suite in our evaluation.
We compare against CRA 
(or Compositional Recurrence Analysis \cite{FMCAD:FK2015}), another invariant
generation tool based on analyzing recurrence relations;
VeriAbs 1.4.2, the winner for the
\texttt{ReachSafety} category of SV-COMP;
and Ultimate Automizer 0.2.2, which performed best on the \texttt{nla-digbench} suite.

\subsection{RQ1: How does \solvername{} perform on standard SMT benchmarks?}
\label{sec:smt-benchmark-perf}

Table~\ref{tab:weak-solver-smt} records the results of running the five
solvers including the weak theory solver on nonlinear SMT benchmarks.
Since the reals are a model of $\ThQ$, if a formula is unsatisfiable
modulo $\ThQ$ then it is unsatisfiable modulo \NRA{}, but the converse does not hold
i.e., using a $\ThQ$-solver on NRA tasks can give false SAT results,
but not false UNSAT results).  The same holds for $\ThZ$ and \NIA{}.


\solvername{} does not appear to be competitive for either
\texttt{QF_NRA} or \texttt{QF_NIA}.
It proves UNSAT for
$6$ out of $40$ and $15$ out of $41$ tasks in the \texttt{QF_NRA} and
\texttt{QF_NIA} suites, respectively, much lower than other SMT solvers we
compared against.
\solvername{} reports false positives on
$23$ out of $40$ \textit{UNSAT} tasks 
in \texttt{QF_NRA} and $15$ out of $41$ \textit{UNSAT} tasks in \texttt{QF_NIA}.
It performs particularly poorly on crafted benchmark suites, in which
there is often a need to reason about the interaction between multiplication
and the order relation, which is not axiomatized by our theories.  It performs particularly well on
verification tasks (e.g., the \texttt{hycomp} suite in \texttt{QF_NRA}, and the
\texttt{LassoRanker} and \texttt{UltimateLassoRanker} suites in \texttt{QF_NIA}),
but all other solvers we tested also performed well on these tasks.
In terms of running time and number of timeouts,
\solvername{} is competitive with other solvers.
We note that there is substantial room for improving the performance of \solvername{};
in particular, it does not tightly integrate the theory solver with the underlying SAT solver as
do most modern SMT solvers.

The experimental results suggest that the theories of linear integer / real rings is
not generically suitable as an alternative to the theory of integers / reals for applications
that only require a SAT or UNSAT answer.

\begin{table}[ht]
\centering
\caption{\label{tab:weak-solver-smt} 
Comparison of \solvername{} with other SMT solvers on SMT-COMP benchmarks.  
The ``\#P'' column denotes the number of proved tasks, and the ``\#E''
column denotes the number of tasks on which a solver times out / runs out of
memory. 
(For UNSAT tasks, \solvername{} may find a $\ThQ{}/\ThZ{}$ model when no NRA/NIA model exists -- these account for the discprency between the number of tasks and \#P+\#E).
The reported time is total aggregate, measured in seconds and averaged across 3 runs.  The maximum standard deviation in runtime across all benchmarks is Chilon: 7.64, Z3: 26.04, MathSat: 0.58, CVC4: 3.89, Yices 2.11.}
\tiny
 \resizebox{\columnwidth}{!}{

\begin{tabular}{@{}ccc|c@{\hspace*{2pt}}c@{\hspace*{2pt}}r|c@{\hspace*{2pt}}c@{\hspace*{2pt}}r|c@{\hspace*{2pt}}c@{\hspace*{2pt}}r|c@{\hspace*{2pt}}c@{\hspace*{2pt}}r|c@{\hspace*{2pt}}c@{\hspace*{2pt}}r@{}}
\toprule
 & & & \multicolumn{3}{c|}{Chilon} & \multicolumn{3}{c|}{Z3} & \multicolumn{3}{c|}{MathSAT} & \multicolumn{3}{c|}{CVC4} & \multicolumn{3}{c}{Yices}\\
suite & label &\#task & \#P  & \#E & time   & \#P & \#E & time &  \#P & \#E & time  & \#P & \#E & time  & \#P & \#E & time \\\midrule
\multirow{3}{*}{QFNRA} & SAT & 38 & - & 10/0 & 1215.9  & 36 & 2/0 & 382.5  & 15 & 21/0 & 2603.8 &  9 & 28/0 & 3406.3  & 29 & 9/0 & 1105.5 \\
& UNSAT & 40 & 6 & 11/0 & 1369.1 & 30 & 9/1 & 1193.7 & 29 & 10/0 & 1297.9 &  31 & 9/0 & 1196.2 & 32 & 8/0 & 1073.5 \\
& ? & 27 & - & 14/2 & 1787.8 & - & 12/0 & 1495.7 &  - & 18/0 & 2225.7  & - & 12/0 & 1838.2  & - & 13/0 & 1645.1 \\
\midrule
\multirow{3}{*}{QFNIA} & SAT & 29 & - & 13/0 & 1595.9  & 26 & 3/0 & 539.7  & 21 & 8/0 & 1067.0 &  17 & 7/0 & 907.3  & 19 & 10/0 & 1211.6 \\
& UNSAT & 41 & 14 & 11/0 & 1347.3  & 36 & 5/0 & 662.9 & 37 & 4/0 & 594.5 & 36 & 1/0 & 252.7 & 36 & 5/0 & 617.2 \\
& ? & 28 & - & 17/0 & 2260.5  & - & 11/0 & 1361.5 &  - & 12/0 & 1550.0  & - & 12/0 & 1467.4 & - & 12/0 & 1635.3 \\
\bottomrule
\end{tabular}

 }
\end{table}

\begin{table}[]
    \centering
    \small
    \caption{
   Comparison of recurrence-based invariant generation schemes.
   The reported time is total aggregate, measured in seconds and averaged across 3 runs,
with the number of timeouts enclosed in parentheses.  The maximum standard deviation in runtime across all benchmarks is Chilon: 2.43, CRA-lin: 0.08, CRA: 0.71.
}
    \label{tab:inv-gen-different-conseq}
\begin{tabular}{@{}lc|c@{}r|c@{}r|c@{}r@{}}
\toprule
 & & \multicolumn{2}{c|}{Chilon-inv} & \multicolumn{2}{c|}{CRA-lin} & \multicolumn{2}{c}{CRA}\\
 & \#tasks & \#correct & time & \#correct & time & \#correct & time\\\midrule
\texttt{linear} & 178 & 117 & 979.4 (5) & 109 & \textbf{807.9} (4) & \textbf{140} & 1258.4 (6)\\
\texttt{nonlinear} & 290 & \textbf{232} & \textbf{2265.6} (15) & 27 & 2621.7 (19) & 90 & 11372.4 (63)\\
\bottomrule
\end{tabular}

\end{table}

\begin{table}[ht]
    \centering
    \caption{
    Improvement by path refinement that make use of invariants generated
    using \solvername{}-enabled consequence-finding, along with
    comparisons with other state-of-the-art tools.
    The reported time is total aggregate, measured in seconds and averaged across 3 runs. The maximum standard deviation in runtime across all benchmarks is Chilon: 2.43, Chilon-inv + Refine: 0.15, UAutomizer: 5.00, VeriAbs: 0.00.
}
    \label{tab:inv-gen-path-ref}
    \resizebox{\columnwidth}{!}{
\begin{tabular}{@{}lc|c@{}r|c@{}r|c@{}r|c@{}r@{}}
\toprule
 & & \multicolumn{2}{c|}{Chilon-inv} & \multicolumn{2}{c|}{Chilon-inv + Refine} & \multicolumn{2}{c|}{UAutomizer} & \multicolumn{2}{c}{VeriAbs}\\
 & \#tasks & \#correct & time & \#correct & time & \#correct & time & \#correct & time\\\midrule
\texttt{linear} & 178 & 117 & \textbf{979.4} (5) & 140 & 1295.1 (5) & 121 & 8873.1 (56) & \textbf{141} & 6479.7 (36)\\
\texttt{nonlinear} & 290 & 232 & \textbf{2265.6} (15) & \textbf{233} & 2352.2 (16) & 183 & 16611.6 (107) & 172 & 17375.4 (118)\\
\bottomrule
\end{tabular}

}
\end{table}

\subsection{RQ2: Does Consequence-Finding Extract Useful Consequences That Benefit Client Applications?}
\label{sec:invgen-benchmark-perf}
In this subsection, we evaluate the ability of our consequence-finding procedure
to find ``useful'' consequences.
We evaluate this via the invariant generation scheme described in
Section~\ref{sec:invgen}, implemented in the tool
\solvername{}-inv.
We compare with two similar
recurrence-based invariant generation techniques, both of which
rely on consequence-finding:
CRA-lin \cite{FMCAD:FK2015} uses a complete consequence-finding
procedure modulo linear arithmetic; and CRA \cite{PACMPL:KCBR18} uses
a heuristic consequence-finding procedure modulo non-linear arithmetic.
We also evaluate a second configuration of
\solvername{}-inv, which uses a refinement algorithm from
\cite{POPL:KBCR2019Ref} to improve analysis precision; this is guaranteed
because the analysis that \solvername{}-inv implements is monotone.



\autoref{tab:inv-gen-different-conseq} compares the performance of
invariant generation using \solvername{} with other methods that utilize
consequence-finding.
Results show that \solvername{}-inv indeed performs strictly better 
than CRA-lin on both suites, as expected, since the consequence-finding
in \solvername{}-inv should be more powerful than that of CRA-lin,
and the class of recurrences it considers is more general.
The difference in capability is quite small for the \texttt{linear} suite
but substantial for the \texttt{nonlinear} suite.
The heuristics implemented in CRA are effective for linear
invariant generation, making it the most powerful tool among the three.
The heuristics are less effective for \texttt{nonlinear} suite,
and \solvername{}-inv performs the best.
This experiment demonstrates that the complete consequence-finding procedure
modulo theory $\ThZ$ indeed yields useful consequences for an
invariant generation scheme, especially in contexts where 
non-linear invariants are required.

\autoref{tab:inv-gen-path-ref} provides context by comparing with state-of-the-art solvers.
VeriAbs is a portfolio verifier that employs a variety of techniques,
such as bounded model checking, $k$-induction, and loop summarization~\cite{VERIABS-SVCOMP21}.
It can in particular summarize a numeric loop by accelerating
linear recurrences involving only constants or variables that are unmodified within the
loop~\cite{Darke2015}.
As mentioned in Section~\ref{sec:linear-inv},
our scheme generalizes this (and CRA) to polynomial invariants.
Ultimate Automizer implements a counterexample-guided abstraction refinement
algorithm following the trace abstraction paradigm~\cite{Heizmann2009}.
The invariants it finds depend on the properties of interest
(the assertions to be proved) and the algorithm used to compute interpolants~\cite{UAUTOMIZER-COMP18}.
\solvername{}-inv performs 
well on both suites compared to these 
tools. 
With path refinement,
\solvername{}-inv performs even better (e.g., by enabling it to find disjunctive invariants), at a cost of time overhead.

\nic{Maybe a bit more detail about \solvername{}-inv? In particular,
  does this mean using path refinement fed back into \solvername{} for invariant generation,
  which then generate better invariants for path refinement again?
}
\zak{Refinement is a domain construction, which takes a star operator as input and returns another as output.  Iterating the construction is a no-op.}

%% file: related.tex
\section{Related Work} \label{sec:relwork}

\paragraph{Real arithmetic}
The decidability of the theory of real closed fields is a classical result due to \citet{Tarski1949} and \citet{Seidenberg1954}.  Practical (complete) algorithms for this theory are based on cylindrical algebraic decomposition (CAD) \cite{Collins1975, IJCAR:JdeM12, JSC:KA2020}.  Due to the high computational complexity of decision procedures for the reals, a number of (incomplete) heuristic techniques have been devised
\cite{CSL:Tiwari2005,CAV:TL2014,LPAR:ZM2010}.  Similar to our approach,  \citet{CSL:Tiwari2005}'s method combines techniques from Gr\"{o}bner bases and linear programming; however, the result of the combination is a semi-decision procedure for the existential theory of the reals, whereas we achieve a decision procedure for a weaker theory.

The \textit{$\delta$-complete} decision procedure for the existential theory of the reals presented in \cite{IJCAR:GAC2012} is similar in spirit to our work, in that the method gives up completeness in the classical sense while retaining a weaker version of it.  Rather than using a standard model of the reals and relaxing the definition of satisfaction (each constraint is ``within $\delta$'' of being satisfied), we take the approach of using classical first-order logic, but admit non-standard models.

Positivestellens\"{a}tze are a class of theorems that characterize sets of positive polynomial consequences of a system of inequalities over the reals (or a real closed field)~\cite{krivine-positivstellensatz}.
Lemma~\ref{lem:consequence-cone} might be thought of as an analogue of a Positivestellensatz for
$\ThQ$.
Putinar's Positivestellensatz \cite{Putinar1993} bears particular resemblance to our results: it asserts that the entailment
$\bigwedge_{z \in Z} z = 0 \land \bigwedge_{p \in P} p \geq 0 \models_{\NRA} q \geq 0$ holds exactly when $q$ is in the \textit{quadratic module} generated by $Z$ and $P$ (provided that certain technical restrictions implying compactness are satisfied).
The quadratic module generated by $Z$ and $P$ is
$\qideal{Z} + \product{\Sigma^2[X]}{P \cup \set{1}}$ (where $\Sigma^2[X]$ is the set of sum-of-squares polynomials over $X$), mirroring the structure of algebraic cones, but with sum-of-squares polynomials in place of non-negative rationals.
Every quadratic module over $\mathbb{R}[X]$ is also regular algebraic cone (following from the fact that its additive units form an ideal \cite[Prop 2.1.2]{Book:Marshall2008}) (but not vice versa).


\paragraph{Integer arithmetic}
Unlike the case of the reals, the theory consisting of $\sigma_{or}$-sentences that hold over the integers is undecidable (in fact, not even recursively axiomatizable).
However, a number of effective heuristics have been proposed, including combining real relaxation with branch-and-bound \cite{CASC:KCA2016,VMCAI:Jovanovic2017},
bit-blasting \cite{SAT:FGMSTZ2007}, and linearization \cite{CADE:BLNRR2009,TCL:BLROR2019}.
Linearization shares the idea of using linear arithmetic to reason about non-linear arithmetic.  However, our approaches are quite different: \citet{CADE:BLNRR2009,TCL:BLROR2019} essentially restrict the domain of constant symbols to finite ranges (making the approach sound but incomplete for satisfiability), whereas our approach is sound but incomplete for validity.

\paragraph{Non-linear invariant generation}
There are several abstract domains that are capable of representing conjunction of polynomial inequalities \cite{PACMPL:KCBR18,SAS:Colon2004,CAV:GG2008,SAS:BRZ2005}. Such domains incorporate ``best effort'' techniques for reasoning about non-linear arithmetic.  Notably, \citet{PACMPL:KCBR18} and \citet{SAS:BRZ2005} combine techniques from commutative algebra and polyhedral theory.  Our approach differs in that we designed \textit{complete} inference, albeit modulo a weak theory of arithmetic.

There is a line of work on \textit{complete} algorithms for finding invariant polynomial equations of (restricted) loops \cite{ISAAC:RCK2004,TACAS:Kovacs2008,VMCAI:HJK2018,LICS:HOPW2018}.
Another approach is to reduce polynomial invariant generation to linear invariant generation by introducing new dimensions, which provides completeness up to a degree-bound
\cite{POPL:MOS04,ATVA:OBP2016}.
    \citet{PLDI:CFGG2020} obtains a completeness (up to technical parameters) result for template-based generation of invariant polynomial inequalities, based on a ``bounded'' version of Putinar's Positivestellensatz.
Lemma~\ref{lem:star-completeness} is a kind of completeness result, which is relative to a class of recurrences, rather than an invariant ``shape.''



\paragraph{Consequence-finding}
A key feature of the arithmetic theories introduced in this paper is that they enable complete methods for finding and manipulating the set of consequences of a formula (of a particular form).  In the setting of abstract interpretation, this problem is known as \textit{symbolic abstraction} \cite{VMCAI:RSY2004,Thesis:Thakur14}, and is phrased as the problem of computing the (ideally, best) approximation of a formula within some abstract domain.  Symbolic abstraction algorithms are known for
predicate abstraction \cite{CAV:GS1997},
equations \cite{JAR:BB2014},
affine equations \cite{VMCAI:RSY2004, ENTCS:TLLR13},
template constraint domains (intervals, octagons, etc) \cite{POPL:LAKGC2014},
and convex polyhedra \cite{FMCAD:FK2015}.
\citet{PACMPL:KCBR18} gives a symbolic abstraction procedure for the wedge domain (which can express polynomial inequalities); this procedure is ``best effort,'' whereas our consequence-finding algorithm offers completeness guarantees.

It is interesting to note that our lazy consequence-finding algorithm (Algorithm~\ref{alg:consequence}) can be seen as an instantiation of \citet{VMCAI:RSY2004}'s symbolic abstraction algorithm.  The termination argument of this algorithm relies on the abstract domain satisfying the ascending chain condition, which algebraic cones do \textit{not}; instead, we exploit the fact that we can compute \textit{minimal} models of $\ThQ$ / $\ThZ$, and each formula has only finitely many.

%% file: appendix.tex
\clearpage
\section{Appendix}

This appendix supplements proofs.

\printProofs

\begin{lemma}
  \label{lem:eval-ring-hom}
  Let $X, Y$ be disjoint sets of variables,
  $f: \mathbb{Q}[X, Y] \to \mathbb{Q}[X]$ be a ring homomorphism
  such that $f(x) = x$ for all $x \in X$ and
  $f(y) \in \mathbb{Q}[X]$ for all $y \in Y$.
  Let $T = \set{y - f(y) : y \in Y}$.
  Then for all $p \in \mathbb{Q}[X, Y]$,
  $p - f(p) \in \product{\mathbb{Q}[X, Y]}{T}$.
\end{lemma}
\begin{proof}
  Follows by noting that $f$ is a ring homomorphism and hence a congruence on terms.
\end{proof}